\documentclass[11pt,a4paper]{article}

\usepackage[T1]{fontenc}
\usepackage[utf8]{inputenc}
\usepackage{lmodern}

\usepackage{amsmath}
\usepackage{amssymb}
\usepackage{amsthm}
\usepackage{mathtools}
\usepackage{mathrsfs}

\usepackage{booktabs}
\usepackage{enumitem}
\setlist[enumerate]{leftmargin=*, itemsep=0.25em}
\usepackage{tikz-cd}

\usepackage{graphicx}
\usepackage{xcolor}
\usepackage[numbers,sort&compress]{natbib}
\usepackage[
  colorlinks=true,
  linkcolor=blue,
  citecolor=blue,
  urlcolor=blue
]{hyperref}

\usepackage[a4paper,margin=1in]{geometry}
\theoremstyle{plain}
\newtheorem{theorem}{Theorem}[section]
\newtheorem{proposition}[theorem]{Proposition}

\newtheorem{corollary}[theorem]{Corollary}

\theoremstyle{definition}
\newtheorem{definition}[theorem]{Definition}

\theoremstyle{remark}
\newtheorem{remark}[theorem]{Remark}

\newcommand{\bbC}{\mathbb{C}}
\newcommand{\bbR}{\mathbb{R}}
\newcommand{\bbZ}{\mathbb{Z}}
\newcommand{\bbQ}{\mathbb{Q}}
\newcommand{\bbT}{\mathbb{T}}

\newcommand{\calE}{\mathcal{E}}
\newcommand{\calF}{\mathcal{F}}

\newcommand{\calN}{\mathcal{N}}

\newcommand{\dd}{\mathrm{d}}

\newcommand{\bulk}{\mathrm{bulk}}
\newcommand{\link}{\mathrm{link}}
\newcommand{\src}{\mathrm{src}}
\newcommand{\orb}{\mathrm{orb}}

\DeclareMathOperator{\ch}{ch}
\DeclareMathOperator{\Ch}{Ch}

\DeclareMathOperator{\rank}{rank}
\DeclareMathOperator{\Res}{Res}

\newcommand{\pair}[2]{\left\langle #1,#2\right\rangle}

\newcommand{\Mfive}{\mathrm{M5}}
\newcommand{\Xeleven}{X^{11}}
\newcommand{\Wsix}{W^6}
\newcommand{\Normal}{N_{W/X}}

\newcommand{\Slink}{S^4_{\link}}

\newcommand{\Porb}{\mathbb{P}_{\orb}}

\title{\bfseries M5-brane source charges, Chern--Connes pairings and noncommutative compactifications}

\author{
Ioannis P. ZOIS\\
Exeter College, Oxford\\
OX1 3DP, United Kingdom\\
\texttt{i.zois@exeter.oxon.org}
}

\date{}

\begin{document}

\maketitle

\begin{abstract}
We formulate a source-adapted Chern--Connes description of M5-brane magnetic charge in physically motivated noncommutative and quotient linking sectors. The starting point is the local codimension-five geometry of an M5-brane source: the magnetic charge is measured on a four-dimensional link. In the smooth local model the link is the ordinary four-sphere, and the pairing reduces to the standard flux integral of the M-theory four-form field strength.

The construction is motivated by the published even-codimensional Noncommutative Foliation Invariant (NCFI for short), where Connes' transverse fundamental cyclic cocycle is paired with an even transverse-geometric \(K\)-class. For an M5-brane the relevant codimension is odd, so the required odd \(K\)-class is supplied instead by the boundary map of the source/linking extension. Noncommutative or groupoid algebras are used only when they arise from physical mechanisms: Matrix/M-theory compactification on noncommutative tori, reduction to D-branes in background \(B\)-fields, large \(C\)-field M5 limits, or quotient/orbifold source links. The result is a topological charge formula, not a derivation of the full M5-brane dynamics or of the complete M-theory action.
\end{abstract}


\tableofcontents
\flushbottom

%
\section{Introduction}
\label{sec:introduction}

M-theory contains extended magnetic sources for the three-form \(C\)-field. In the
low-energy eleven-dimensional description, the \(C\)-field has field strength \(G_4\), and
an M5-brane is a magnetic source for \(G_4\) \cite{CremmerJuliaScherk,WittenStringDynamics,
WittenFlux,WittenFivebrane}. If
\[
  \Wsix\hookrightarrow \Xeleven
\]
is an M5-brane worldvolume, then its normal bundle has rank five. Locally the transverse
normal geometry is modelled by \(\bbR^5\), and a small linking cycle around the source is
a four-sphere,
\[
  \Slink\simeq S^4.
\]
The local magnetic charge is measured by
\[
  Q_{\Mfive}
  =
  \int_{\Slink}\frac{G_4}{2\pi}.
\]
The elementary-looking expression above is already topological: it is a linking charge
around a codimension-five source. It is also not the whole story, because the M-theory
\(C\)-field obeys shifted quantization and has a refined global description involving
differential cohomology, \(E_8\)-model data, and anomaly constraints
\cite{WittenFlux,DiaconescuFreedMoore,HopkinsSinger,FreedHarveyMinasianMoore}.

The aim of this paper is to express such source charges by Chern--Connes pairings when
the physical linking or compactification sector is naturally represented by a
noncommutative or groupoid algebra. The construction is deliberately narrower than a
proposal for the full M-theory action. It is a topological charge formula. In the smooth
local model the formula reduces to the ordinary M5 magnetic charge. In compactified,
orbifold or flux-deformed sectors the same formalism gives a source-adapted pairing
between a cyclic cocycle representing the linking fundamental class and a \(K\)-theory
class representing the charge sector.

The mathematical starting point is the Noncommutative Foliation Invariant (NCFI for short) introduced
in \cite{ZoisSigma} for \emph{even} codimensional foliated manifolds. The construction relies heavily on noncommutative geometry techniques but the initial motivation came from physics, in particular from non-linear $\sigma$-models. In that published
construction, for a transversally oriented foliation of even codimension \(q=2n\), one
uses the holonomy groupoid \(G_{\calF}\), a smooth convolution algebra
\[
  A^\infty_{\calF}=C_c^\infty(G_{\calF},\Omega^{1/2}),
\]
and the reduced foliation \(C^*\)-algebra
\[
  A_{\calF}=C_r^*(G_{\calF}).
\]
Connes' transverse fundamental cyclic cocycle gives an even class
\[
  \phi_{\calF}\in HC^{2n}(A^\infty_{\calF}),
\]
while the transverse geometry of the foliation supplies an even \(K\)-class
\[
  [e_{\calF}]\in K_0(A_{\calF}).
\]
The invariant in even codimension is the Chern--Connes pairing
\[
  Z(\calF)=\pair{\phi_{\calF}}{[e_{\calF}]}.
\]
For the construction and its nonlinear \(\sigma\)-model motivation, see
\cite{ZoisSigma}; for the transverse fundamental cyclic cocycle and the general
Chern--Connes framework, see \cite{ConnesTransverse,ConnesNCG,LodayCyclic,
KhalkhaliNCG}. An earlier eleven-dimensional motivation based on secondary
characteristic classes of flat foliations appeared in \cite{ZoisMTheory}.

The case relevant to an M5-brane is different. A codimension-five source is odd
codimensional. In odd codimension \(q=2n+1\), the transverse cyclic cocycle has odd
degree,
\[
  \phi_{\calF}\in HC^{2n+1}(A^\infty_{\calF}),
\]
and therefore pairs with \(K_1(A_{\calF})\), not with \(K_0(A_{\calF})\). Thus an odd
codimensional numerical pairing requires an additional odd \(K\)-class,
\[
  [u]\in K_1(A_{\calF}).
\]
There is no canonical functorial map
\[
  K_0(A_{\calF})\longrightarrow K_1(A_{\calF})
\]
for a general \(C^*\)-algebra. Bott periodicity shifts parity only after suspension, and
cyclic periodicity shifts degree by two. Hence the odd \(K\)-class must come from
additional physical, geometric, dynamical, analytic or extension data.

For the M5-brane source problem, the required odd class is supplied by the
source/linking extension. The physically transparent object is first an even linking
class
\[
  x_{\check C}\in K_0(A_{\link}),
\]
where \(A_{\link}\) is the algebra of the four-dimensional link and \(x_{\check C}\) is the
\(K\)-theoretic representative, at the Chern-character level, of the restricted \(C\)-field
flux. The source neighbourhood then gives a boundary map
\[
  \partial:K_0(A_{\link})\longrightarrow K_1(A_{\bulk}),
\]
and the odd class is
\[
  [u_{\check C}]:=\partial x_{\check C}.
\]
Thus in the codimension-five case the odd \(K\)-class is not chosen as an unrelated
generator. It is forced by the boundary of the source/linking geometry.

A central point of this paper is that the noncommutative or groupoid algebras used below
are not introduced arbitrarily. They must have a physical origin. The first mechanism is
Matrix/M-theory compactification. Connes, Douglas and Schwarz showed that
noncommutative tori arise naturally in Matrix-theory compactifications on tori with
background three-form data \cite{ConnesDouglasSchwarz}. This places noncommutative
torus algebras directly inside the M-theory compactification problem, in the same period
as the BFSS Matrix-theory proposal \cite{BanksFischlerShenkerSusskind}. Morita
equivalence of noncommutative tori is then interpreted as a duality compatibility
condition \cite{SchwarzMorita,RieffelSchwarz}.

The second mechanism comes from reduction to D-brane systems and from large
background \(C\)-field limits. Douglas and Hull showed how noncommutative tori appear
for D-branes in string compactifications \cite{DouglasHull}. The Seiberg--Witten analysis
then clarified the relation between ordinary and noncommutative gauge descriptions in a
background \(B\)-field \cite{SeibergWitten}, with related deformation-quantization
interpretations in \cite{Schomerus}. On the M-theory side, M5-branes in large constant
\(C\)-field backgrounds lead to noncommutative open-membrane or Nambu--Poisson type
structures \cite{BergshoeffNoncommM5,BergshoeffCriticalFields,HoImamuraMatsuoShiba,
HoLargeCFieldReview,ChenFuruuchiHoTakimi}. These works do not imply that every M5
link should be replaced by an associative noncommutative algebra. They do justify the
more careful statement used here: when the physical compactification or charge sector is
represented by such an algebra, its source charge should be expressible by a
Chern--Connes pairing.

The third mechanism is quotient geometry. Branes at orbifold or quotient singularities
are standard in string/M-theory. Witten studied five-branes in M-theory on an orbifold
\cite{WittenOrbifoldFivebranes}; Douglas and Moore used D-branes at ALE orbifold
singularities to obtain quiver gauge theories \cite{DouglasMooreOrbifolds}; and ABJM
theory describes M2-branes probing \(\bbC^4/\bbZ_k\), with a large-\(N\) M-theory dual
on \(AdS_4\times S^7/\bbZ_k\) \cite{ABJM}. In such settings the quotient has stabilizer
data, and the correct geometric object is naturally a groupoid or stack rather than an
ordinary manifold. Orbifold groupoids therefore enter the present paper only as models
of physical quotient or source links, not as arbitrary transverse replacements.

The basic object studied below is a source-adapted Chern--Connes charge
\[
  Q_{\src}(G_{\link};\check C)
  =
  \pair{\tau_{\link}}{x_{\check C}}.
\]
Here \(G_{\link}\) is a physically admissible linking groupoid or compactification algebra,
\[
  A_{\link}=C_r^*(G_{\link})
\]
is its \(C^*\)-algebra, \(\tau_{\link}\) is a cyclic cocycle representing the linking
fundamental class, and \(x_{\check C}\in K_0(A_{\link})\) is the \(K\)-theoretic charge
class associated with the \(C\)-field flux at the level of the Chern character. In the
ordinary smooth local model,
\[
  A_{\link}=C(S^4),
\]
and the formula reduces to
\[
  \pair{\tau_{S^4}}{x_{\check C}}
  =
  \int_{S^4}\frac{G_4}{2\pi}.
\]

The construction should be read with two limitations. First, complex cyclic cohomology
sees the de Rham or Chern-character image of a charge. It does not by itself see torsion,
quadratic refinements, or global anomaly phases. Those require differential cohomology,
generalized cohomology, or related refinements
\cite{WittenFlux,DiaconescuFreedMoore,HopkinsSinger,MoorePageCharges}. Second, the
pairing computes a linking charge, not the full dynamics of the M5-brane or the complete
anomaly polynomial of the six-dimensional \((2,0)\) theory. Fivebrane anomaly
cancellation and the self-dual tensor field impose further constraints
\cite{WittenFivebrane,FreedHarveyMinasianMoore}.

The swampland perspective imposes another restriction. A source charge in quantum
gravity should not be interpreted as a new exact global superselection label. Vafa's
swampland proposal \cite{VafaSwampland}, the Weak Gravity Conjecture
\cite{ArkaniHamedMotlNicolisVafaWGC}, the absence of global symmetries in quantum
gravity \cite{BanksSeiberg,HarlowOoguriConstraints,HarlowOoguriSymmetries}, and the
cobordism swampland viewpoint \cite{McNamaraVafaCobordism} all suggest that any
topological charge must be attached to dynamical objects or be trivialized by allowed
defects. The source-adapted pairing considered here is therefore proposed as the
Chern--Connes shadow of a brane/source charge, not as an independent conserved global
label.

The contributions of the paper are the following.

\begin{enumerate}[label=(\roman*)]
  \item We distinguish the published even-codimensional NCFI from the M5 source
  problem. The published construction supplies the Chern--Connes pairing framework
  \cite{ZoisSigma}; the codimension-five case requires an additional odd
  \(K\)-class.

  \item We identify the boundary/source extension as the relevant origin of this odd
  \(K\)-class in the M5 problem:
  \[
    [u_{\check C}]=\partial x_{\check C}.
  \]

  \item We prove the local commutative normalization theorem: for the smooth link
  \(S^4\), the source-adapted Chern--Connes pairing reproduces the ordinary M5
  magnetic charge.

  \item We formulate the compactified version on noncommutative four-tori in the
  Connes--Douglas--Schwarz setting and impose Morita invariance as the duality
  compatibility condition.

  \item We treat orbifold links as physically admissible quotient links and compute a
  weighted projective example as a controlled groupoid calculation.

  \item We state the admissibility constraints coming from shifted \(C\)-field
  quantization, M5 anomaly inflow and swampland consistency.
\end{enumerate}

The structure of the paper is as follows. Section~\ref{sec:physical-origin} explains the
physical origin of the linking algebras: smooth links, noncommutative tori from
Matrix/M-theory compactification, \(B\)-field and large \(C\)-field mechanisms, and
orbifold source links. Section~\ref{sec:published-ncfi-to-source} recalls the published
even-codimensional NCFI construction and explains why the odd codimension-five case
requires a boundary-supplied \(K_1\)-class. Section~\ref{sec:local-m5} proves the local
M5 linking theorem. Section~\ref{sec:cds} gives the Connes--Douglas--Schwarz
noncommutative-torus version. Section~\ref{sec:orbifold-link} treats the orbifold linking
example. Section~\ref{sec:cfield-anomalies} discusses shifted \(C\)-field quantization,
anomaly inflow and swampland constraints. Section~\ref{sec:conclusion} summarizes the
results and open problems.

\section{Physical origin of the linking algebras}
\label{sec:physical-origin}

The construction in this paper is not based on replacing a smooth transverse link by an
arbitrary noncommutative algebra. Such a replacement would be mathematically possible
but physically unmotivated. We therefore impose a physical admissibility principle: the
linking algebra must arise from the M5-brane source geometry, from a compactification
sector, from a background flux limit, or from a physical quotient singularity.

The basic local M5-brane link is smooth and commutative. Noncommutative or groupoid
algebras enter only when the physical background supplies them. The three mechanisms
used in this paper are:
\[
  \begin{gathered}
  \text{Matrix/M-theory compactification},\\
  B\text{- or }C\text{-field brane limits},\\
  \text{orbifold or quotient source geometry}.
  \end{gathered}
\]
This section fixes the physical meaning of these mechanisms before any groupoid or
noncommutative extension is introduced.

\subsection{The smooth local M5 source}

Let
\[
  \Wsix\hookrightarrow \Xeleven
\]
be an M5-brane worldvolume. The normal bundle
\[
  \Normal\longrightarrow \Wsix
\]
has rank five. Locally, near a point of \(W\), the geometry is modelled by
\[
  X_{\mathrm{loc}}=W^6\times\bbR^5,
  \qquad
  W_{\mathrm{loc}}=W^6\times\{0\}.
\]
A small linking cycle around the source is
\[
  S^4_{\link}\subset \bbR^5.
\]
The M5-brane is magnetically charged under the M-theory \(C\)-field, whose field strength
is \(G_4\). In the presence of a magnetic source, the Bianchi identity has the schematic
current form
\[
  \dd G_4=2\pi N\,\delta_W,
\]
and the local magnetic charge is measured by the linking integral
\[
  Q_{\Mfive}
  =
  \int_{S^4_{\link}}\frac{G_4}{2\pi}.
\]
This is the standard brane-charge situation: a magnetic source of codimension five is
measured on a linking four-cycle. The general \(p\)-brane charge viewpoint is part of the
standard M-theory and brane literature \cite{TownsendPBranes,DuffSupermembranes,
StromingerOpenPBranes,GuevenBlackPBranes,WittenFivebrane}.

In this local model there is no noncommutativity. The linking algebra is simply
\[
  A_{\link}=C(S^4),
  \qquad
  A_{\link}^{\infty}=C^\infty(S^4),
\]
and the fundamental cyclic cocycle is the ordinary de Rham cyclic cocycle of the
oriented four-sphere. The local theorem later in the paper proves that the corresponding
Chern--Connes pairing reproduces
\[
  \int_{S^4}\frac{G_4}{2\pi}.
\]
This smooth model is the normalization test for all later extensions.

\subsection{Noncommutative tori from Matrix/M-theory compactification}

The first physical source of noncommutative algebras is toroidal compactification in
Matrix/M-theory. The BFSS proposal identifies M-theory, in a suitable infinite-momentum
or DLCQ regime, with a large-\(N\) matrix quantum mechanics \cite{BanksFischlerShenkerSusskind}.
Connes, Douglas and Schwarz then showed that compactifications of Matrix theory on
tori naturally lead to noncommutative tori and argued that these backgrounds correspond
in supergravity language to tori with constant background three-form data
\cite{ConnesDouglasSchwarz}.

Thus the noncommutative torus
\[
  A_\Theta=C(\bbT^d_\Theta)
\]
is not a formal replacement for a commutative torus. It is a compactification algebra
arising from the physical background. Its generators
\[
  U_1,\ldots,U_d
\]
satisfy
\[
  U_iU_j=e^{2\pi i\Theta_{ij}}U_jU_i,
  \qquad
  \Theta_{ij}=-\Theta_{ji}.
\]
The deformation parameter \(\Theta\) records compactification data. The smooth algebra
\[
  A_\Theta^\infty
\]
is the natural domain for derivations, traces and cyclic cocycles.

For the present M5 source problem, the dimensionally natural compactified charge sector
is four-dimensional:
\[
  A_\Theta^\infty=C^\infty(\bbT^4_\Theta).
\]
It carries a fundamental cyclic four-cocycle
\[
  \tau_\Theta^{(4)}\in HC^4(A_\Theta^\infty).
\]
Given a \(K\)-class
\[
  x_{\check C,\Theta}\in K_0(A_\Theta)
\]
representing the compactified \(C\)-field charge sector at the level of the Chern
character, the source charge is written as
\[
  Q_{\src,\Theta}
  :=
  \pair{\tau_\Theta^{(4)}}{x_{\check C,\Theta}}.
\]
In the commutative limit
\[
  \Theta=0,
\]
this reduces to the ordinary Chern-character pairing over \(T^4\).

Morita equivalence is essential in this compactified setting. Schwarz interpreted Morita
equivalence of noncommutative tori as a duality relation in Matrix compactification
\cite{SchwarzMorita}, and Rieffel--Schwarz gave the corresponding multidimensional
Morita-equivalence framework \cite{RieffelSchwarz}. Therefore the charge formula must
be invariant under simultaneous Morita transport of the cyclic cocycle and the
\(K\)-class:
\[
  \pair{\tau_\Theta^{(4)}}{x_\Theta}
  =
  \pair{\tau_{\Theta'}^{(4)}}{x_{\Theta'}}
\]
whenever
\[
  A_\Theta\sim_M A_{\Theta'}
\]
and the pair
\[
  \bigl(\tau_\Theta^{(4)},x_\Theta\bigr)
\]
is transported through the equivalence bimodule. This is the compactification-duality
test for the source charge.

\subsection{D-branes, \texorpdfstring{\(B\)}{B}-fields and noncommutative gauge theory}

A second physical route to noncommutative algebras is obtained by reducing M-theory to
type IIA string theory. Under reduction along the M-theory circle, suitable components
of the M-theory \(C\)-field become the NS two-form \(B\)-field. D-branes in a background
\(B\)-field have noncommutative worldvolume descriptions.

Douglas and Hull related D-branes and noncommutative tori in string compactifications
\cite{DouglasHull}. Chu and Ho derived noncommuting open-string endpoint coordinates
for open strings ending on D-branes in a background \(B\)-field
\cite{ChuHoNoncommOpenString}. Ardalan, Arfaei and Sheikh-Jabbari gave another
string/brane derivation of noncommutative geometry in this setting
\cite{ArdalanArfaeiSheikhJabbari}. Seiberg and Witten then clarified the relation
between ordinary and noncommutative gauge descriptions and the corresponding
low-energy limit \cite{SeibergWitten}. Related noncommutative gauge-theory and
supergravity backgrounds were studied in
\cite{Schomerus,HashimotoItzhaki,MaldacenaRusso}.

This body of work supplies a clear physical meaning to algebras such as
\[
  A_\Theta^\infty.
\]
They are not merely noncommutative analogues of function algebras; they arise as
effective algebras of open-string or brane degrees of freedom in background flux limits.

For the M5-brane charge considered here, the relevance is indirect but important. If an
M5-brane configuration is reduced or compactified to a D-brane system with a background
\(B\)-field, then the corresponding source charge should be expressible in the
noncommutative algebra of the reduced brane system. The Chern--Connes pairing is then
the natural noncommutative replacement of the ordinary Chern-character integral.

\subsection{Large \texorpdfstring{\(C\)}{C}-field M5-branes and open membranes}

There is also an intrinsically M-theoretic route. M5-branes in large or near-critical
\(C\)-field backgrounds lead to open-membrane and noncommutative or higher-Poisson
structures. Bergshoeff, Berman, van der Schaar and Sundell studied noncommutative
M-theory five-branes and critical fields on the M5-brane
\cite{BergshoeffNoncommM5,BergshoeffCriticalFields}. Gopakumar, Minwalla, Seiberg
and Strominger formulated OM theory as a decoupled open-membrane theory associated
with an M5-brane in a near-critical three-form field strength
\cite{GopakumarMinwallaSeibergStromingerOM}. Related duality descriptions of
noncommutative gauge theories and open-string theories were developed in
\cite{GopakumarMaldacenaMinwallaStromingerSDuality}.

More recent Nambu--Poisson descriptions of M5-branes in large \(C\)-field backgrounds
make explicit that the algebraic structure need not be an ordinary associative
noncommutative algebra. Ho, Imamura, Matsuo and Shiba constructed an M5-brane model
in a three-form flux background \cite{HoImamuraMatsuoShiba}; Ho's review summarizes
the large \(C\)-field picture \cite{HoLargeCFieldReview}; and Chen, Furuuchi, Ho and
Takimi discuss the scaling limit, background independence and the Seiberg--Witten map
for the Nambu--Poisson M5-brane theory \cite{ChenFuruuchiHoTakimi}.

This point is important for the present article. The Chern--Connes formula below uses
ordinary associative algebras and their \(K\)-theory. It therefore captures the associative
noncommutative or groupoid-algebra shadow of the source charge. It does not claim to
describe the full higher-algebraic structure of the M5-brane in an arbitrary large
\(C\)-field background. In regimes where the correct structure is genuinely
Nambu--Poisson or higher-categorical, the present pairing should be regarded as a
secondary or decategorified shadow, not as the full theory.

\subsection{Orbifold and quotient source links}

A third physical mechanism is quotient geometry. Branes at orbifold singularities are
standard objects in string/M-theory. Witten studied five-branes in M-theory on an
orbifold \cite{WittenOrbifoldFivebranes}. Douglas and Moore showed how D-branes at
ALE orbifold singularities lead to quiver gauge theories \cite{DouglasMooreOrbifolds}.
ABJM theory provides a central later example: M2-branes probing
\[
  \bbC^4/\bbZ_k
\]
are described by a three-dimensional Chern--Simons--matter theory with an M-theory
dual on
\[
  AdS_4\times S^7/\bbZ_k
\]
at large \(N\) \cite{ABJM}.

In quotient geometries, the naive quotient space discards stabilizer data. The physically
correct object is instead an orbifold, stack or groupoid. If a compact group or finite
group \(\Gamma\) acts on a link \(L\), the quotient link is represented by the action
groupoid
\[
  L\rtimes\Gamma.
\]
Its reduced groupoid \(C^*\)-algebra
\[
  C_r^*(L\rtimes\Gamma)
\]
retains the isotropy data. The corresponding cyclic cocycle is the orbifold or groupoid
fundamental cyclic cocycle, and Chern--Connes pairings become orbifold
characteristic-number pairings. The groupoid formulation of orbifolds is standard
\cite{SatakeVManifolds,KawasakiIndex,RenaultGroupoids,MoerdijkPronk,
CrainicMoerdijkFoliation}.

The weighted projective example later in the paper is used in precisely this sense. The
link
\[
  \Porb(1,2,3)
\]
is a controlled four-dimensional orbifold link presented by a weighted Hopf action
groupoid. It is not asserted to be a universal local model for M5-branes. It is a test case
showing how a physically admissible quotient link is treated by the same source-pairing
formalism.

\subsection{What is not being assumed}

The preceding discussion excludes a common but dangerous misreading. We do not assume
that every formal transverse groupoid defines a physical M5-brane background. A
groupoid may be mathematically natural but physically irrelevant unless it is generated
by a source, compactification, quotient, duality or background flux mechanism.

In particular, non-Hausdorff groupoids are not automatically part of the physical claim.
They may arise in genuine quotient or gluing geometries, and the groupoid
\(C^*\)-algebra then becomes the correct replacement for the pathological quotient. But
without a concrete physical construction, a non-Hausdorff leaf space is only a
mathematical generalization \cite{TuNonHausdorff}. The main results of this paper use smoother and more
controlled situations: \(S^4\) links, noncommutative tori from compactification, and
proper orbifold links.

Similarly, the noncommutativity parameter
\[
  \Theta
\]
of a noncommutative torus should not be identified with the M5 magnetic charge. The
charge is the value of a pairing,
\[
  \pair{\tau_{\link}}{x_{\check C}},
\]
whereas \(\Theta\) is compactification or background data entering the algebra. The two
belong to the same physical \(C\)-field sector, but they are not the same object.

\subsection{Physically admissible linking algebras}

We record the admissibility condition used in the rest of the paper.

\begin{definition}
\label{def:physically-admissible-linking-algebra}
A linking algebra \(A_{\link}\) is called \emph{physically admissible} for the
source-charge construction of this paper if it arises from one of the following mechanisms:
\begin{enumerate}[label=(\roman*)]
  \item the ordinary smooth local M5-brane link
  \[
    S^4_{\link},
    \qquad
    A_{\link}=C(S^4);
  \]

  \item a Matrix/M-theory toroidal compactification with background three-form data,
  giving a noncommutative torus algebra
  \[
    A_{\link}=A_\Theta;
  \]

  \item a reduction to a D-brane system with background \(B\)-field, or an M5-brane
  large \(C\)-field regime, for which an associative noncommutative algebra is a justified
  effective description;

  \item a quotient or orbifold source link presented by an action groupoid or proper
  orbifold groupoid,
  \[
    G_{\link}=L\rtimes\Gamma,
    \qquad
    A_{\link}=C_r^*(G_{\link});
  \]

  \item a source-gluing or quotient construction whose non-Hausdorff or groupoid
  character is part of the physical background, rather than an arbitrary mathematical
  replacement.
\end{enumerate}
\end{definition}

For a physically admissible linking algebra, the basic source charge takes the form
\[
  Q_{\src}(G_{\link};\check C)
  =
  \pair{\tau_{\link}}{x_{\check C}},
\]
where
\[
  \tau_{\link}\in HC^4(A_{\link}^{\infty})
\]
is the cyclic cocycle representing the linking fundamental class, and
\[
  x_{\check C}\in K_0(A_{\link})
\]
is the \(K\)-theoretic charge class whose Chern--Connes character represents the
restricted \(C\)-field flux at the level of the Chern character.

The next section explains how this source-charge formulation is related to the published
even-codimensional NCFI construction and why the odd codimension-five case requires the
boundary-supplied class
\[
  [u_{\check C}]=\partial x_{\check C}.
\]

\section{From the published even-codimensional NCFI to the source odd class}
\label{sec:published-ncfi-to-source}

We now separate the published noncommutative foliation invariant from the source-charge
construction used in the rest of the paper. This distinction is necessary. The invariant
introduced in \cite{ZoisSigma} was formulated for
transversally oriented foliations in \emph{even} codimension. The M5-brane problem considered
here is not simply an application of that even-codimensional definition. It is a
codimension-five source problem, and hence it requires an odd \(K\)-class.

The point of this section is therefore twofold. First, we recall exactly which part of the
published even-codimensional NCFI is being used. Second, we define the boundary-source
mechanism which supplies the required odd class in the M5-brane setting.

\subsection{The published even-codimensional construction}

Let
\[
  (M,\calF)
\]
be a compact smooth manifold equipped with a transversally oriented foliation of even
codimension
\[
  q=2n.
\]
Let
\[
  G_{\calF}\rightrightarrows M
\]
be the holonomy groupoid of the foliation. The smooth convolution algebra is
\[
  A_{\calF}^{\infty}:=C_c^\infty(G_{\calF},\Omega^{1/2}),
\]
and the reduced foliation \(C^*\)-algebra is
\[
  A_{\calF}:=C_r^*(G_{\calF}).
\]
Connes' transverse fundamental cyclic cocycle gives a cyclic \(2n\)-cocycle
\[
  \phi_{\calF}\in HC^{2n}(A_{\calF}^{\infty}).
\]
The even-codimensional NCFI introduced in \cite{ZoisSigma} pairs this cocycle with an
even \(K\)-class
\[
  [e_{\calF}]\in K_0(A_{\calF})
\]
defined from the transverse geometry of the foliation. Thus, under the usual smooth
representability hypothesis for the \(K\)-class, the invariant has the form
\[
  Z_{\mathrm{even}}(\calF)
  =
  \pair{\phi_{\calF}}{[e_{\calF}]}.
\]
This is the Chern--Connes pairing of the transverse fundamental cyclic cocycle with the
even transverse-geometric \(K\)-class. The relevant background is Connes' transverse
fundamental class and the general Chern--Connes pairing
\cite{ConnesTransverse,ConnesNCG,LodayCyclic,KhalkhaliNCG}.

This published construction is important here because it supplies the conceptual
template:
\[
  \text{cyclic fundamental class}
  \quad\times\quad
  K\text{-theoretic geometric or charge class}
  \quad\longrightarrow\quad
  \text{number}.
\]
However, the M5-brane source problem uses a different \(K\)-class. The relevant class is
not the transverse-geometric class of a foliation. It is the charge class associated with
the restricted \(C\)-field flux on a linking geometry.

\subsection{The parity obstruction in odd codimension}

Suppose instead that the transverse codimension is odd,
\[
  q=2n+1.
\]
Then Connes' transverse fundamental cyclic cocycle has odd degree:
\[
  \phi_{\calF}\in HC^{2n+1}(A_{\calF}^{\infty}).
\]
The Chern--Connes pairing has parity:
\[
  HC^{2n}(A^\infty)\times K_0(A)\longrightarrow \bbC,
\]
and
\[
  HC^{2n+1}(A^\infty)\times K_1(A)\longrightarrow \bbC.
\]
Therefore an odd cyclic cocycle naturally pairs with a \(K_1\)-class, not with a
\(K_0\)-class.

This is a structural obstruction. For a general \(C^*\)-algebra \(A\), there is no
canonical functorial map
\[
  K_0(A)\longrightarrow K_1(A).
\]
Bott periodicity gives a parity shift only after suspension:
\[
  K_0(A)\simeq K_1(SA),
  \qquad
  SA=C_0((0,1))\otimes A,
\]
not a canonical map from \(K_0(A)\) to \(K_1(A)\) on the same algebra. Similarly,
Connes' periodicity operator in cyclic cohomology changes degree by two and therefore
preserves parity. Hence the odd-codimensional case requires extra data: a distinguished
odd \(K\)-class
\[
  [u]\in K_1(A).
\]

In the present paper the extra data are not imported from an arbitrary foliation model.
They come from the physical source geometry. An M5-brane is a codimension-five
magnetic source. Its charge is first defined on a four-dimensional link, hence by an even
linking class
\[
  x_{\check C}\in K_0(A_{\link}).
\]
The corresponding odd class is then obtained from the boundary map of the source
neighbourhood:
\[
  [u_{\check C}]
  :=
  \partial x_{\check C}.
\]
This is the parity mechanism used throughout the paper.

\subsection{The source and linking extension}

Let
\[
  A_{\link}
\]
be the \(C^*\)-algebra of a physically admissible four-dimensional linking sector, and let
\[
  A_{\link}^{\infty}\subset A_{\link}
\]
be the smooth algebra on which the linking cyclic cocycle is defined. In the simplest
local M5-brane model,
\[
  A_{\link}=C(S^4),
  \qquad
  A_{\link}^{\infty}=C^\infty(S^4).
\]
In compactified or quotient sectors, \(A_{\link}\) may instead be a noncommutative torus
algebra or a groupoid \(C^*\)-algebra, provided it arises from one of the physical
mechanisms described in Section~\ref{sec:physical-origin}.

The source neighbourhood gives an extension
\[
  0
  \longrightarrow
  A_{\bulk}
  \longrightarrow
  A_{\mathrm{src}}
  \longrightarrow
  A_{\link}
  \longrightarrow
  0.
\]
Here \(A_{\bulk}\) is the algebra of the punctured or cone-interior normal sector, and
\(A_{\mathrm{src}}\) is the algebra of the source neighbourhood or cone compactification.
The extension gives a boundary map in \(K\)-theory,
\[
  \partial:K_0(A_{\link})\longrightarrow K_1(A_{\bulk}),
\]
as part of the six-term exact sequence of \(K\)-theory
\cite{BlackadarKTheory,KasparovExtensions}.

Let
\[
  x_{\check C}\in K_0(A_{\link})
\]
be the even linking class associated with the \(C\)-field flux sector at the level of the
Chern character. The source boundary class is then defined by
\[
  [u_{\check C}]
  :=
  \partial x_{\check C}
  \in K_1(A_{\bulk}).
\]
Thus the odd \(K\)-class is not an arbitrary choice. It is the boundary image of the
physical linking charge class.

\begin{definition}
\label{def:source-boundary-datum}
A \emph{source-boundary datum} consists of:
\begin{enumerate}[label=(\roman*)]
  \item a physically admissible linking algebra
  \[
    A_{\link}^{\infty}\subset A_{\link};
  \]

  \item an even linking cyclic cocycle
  \[
    \tau_{\link}\in HC^4(A_{\link}^{\infty});
  \]

  \item a charge class
  \[
    x_{\check C}\in K_0(A_{\link});
  \]

  \item a source extension
  \[
    0
    \longrightarrow
    A_{\bulk}
    \longrightarrow
    A_{\mathrm{src}}
    \longrightarrow
    A_{\link}
    \longrightarrow
    0;
  \]

  \item an odd cyclic cocycle
  \[
    \phi_{\bulk}\in HC^5(A_{\bulk}^{\infty})
  \]
  on a smooth model of the bulk algebra.
\end{enumerate}
The associated boundary odd class is
\[
  [u_{\check C}]
  :=
  \partial x_{\check C}
  \in K_1(A_{\bulk}).
\]
\end{definition}

The even source charge is
\[
  Q_{\src}(A_{\link};\check C)
  :=
  \pair{\tau_{\link}}{x_{\check C}}.
\]
The strict odd boundary expression is
\[
  Z_{\mathrm{bdry}}(A_{\bulk};\check C)
  :=
  \pair{\phi_{\bulk}}{\partial x_{\check C}}.
\]
The two expressions are related when the cyclic cocycles are compatible with the source
extension.

\subsection{Boundary compatibility}

The source extension has two sides. On \(K\)-theory it gives the boundary map
\[
  \partial:K_0(A_{\link})\to K_1(A_{\bulk}).
\]
On the cyclic-cohomology side, a residue, restriction or Chern--Simons transgression
may relate the odd bulk cocycle to the even linking cocycle:
\[
  \Res(\phi_{\bulk})=\tau_{\link}.
\]
This compatibility is a noncommutative version of Stokes' theorem. In periodic cyclic
cohomology, such boundary operations are governed by cyclic excision and
transgression \cite{QuillenCyclic,CuntzQuillenExcision,LodayCyclic}.

\begin{definition}
\label{def:source-boundary-compatible}
A source-boundary datum is called \emph{boundary compatible} if
\[
  \Res(\phi_{\bulk})=\tau_{\link}
\]
and, for the charge classes under consideration,
\[
  \pair{\phi_{\bulk}}{\partial x_{\check C}}
  =
  \varepsilon\,
  \pair{\tau_{\link}}{x_{\check C}},
  \qquad
  \varepsilon\in\{+1,-1\}.
\]
The sign \(\varepsilon\) is fixed by the orientation convention for the boundary of the
source neighbourhood.
\end{definition}

When the orientation convention is chosen so that \(\varepsilon=+1\), the strict odd
codimension-five expression and the even linking expression give the same number:
\[
  Z_{\mathrm{bdry}}(A_{\bulk};\check C)
  =
  Q_{\src}(A_{\link};\check C).
\]
In the local M5 model this is the equality between the odd boundary pairing on the
punctured normal fibre and the even pairing on the linking four-sphere.

\begin{remark}
The pairing
\[
  \pair{\tau_{\link}}{x_{\check C}}
\]
is the physically transparent one, because the M5 magnetic charge is measured on a
four-dimensional link. The odd pairing
\[
  \pair{\phi_{\bulk}}{\partial x_{\check C}}
\]
is the codimension-five boundary form of the same charge. The odd \(K_1\)-class is
therefore determined by the source extension.
\end{remark}

\subsection{The local source extension}

For the local smooth M5-brane model, the normal fibre is \(D^5\) and the link is
\[
  \partial D^5=S^4.
\]
The corresponding commutative extension is
\[
  0
  \longrightarrow
  C_0(\bbR^5)
  \longrightarrow
  C(D^5)
  \longrightarrow
  C(S^4)
  \longrightarrow
  0.
\]
The boundary map is
\[
  \partial:K_0(C(S^4))\longrightarrow K_1(C_0(\bbR^5)).
\]
Thus, if
\[
  x_{\check C}\in K_0(C(S^4))
\]
is the linking class representing the \(G_4\)-flux at the level of the Chern character,
then the associated odd class is
\[
  [u_{\check C}]
  =
  \partial x_{\check C}
  \in K_1(C_0(\bbR^5)).
\]
This is the source of the odd class used in the local theorem.

The smooth algebra on the link is
\[
  C^\infty(S^4),
\]
with fundamental cyclic cocycle
\[
  \tau_{S^4}\in HC^4(C^\infty(S^4)).
\]
The bulk side has an odd cyclic cocycle
\[
  \phi_{\bbR^5}\in HC^5(C_c^\infty(\bbR^5)),
\]
whose boundary residue is normalized to be
\[
  \Res(\phi_{\bbR^5})=\tau_{S^4}.
\]
With the compatible orientation convention, the boundary formula is
\[
  \pair{\phi_{\bbR^5}}{\partial x_{\check C}}
  =
  \pair{\tau_{S^4}}{x_{\check C}}.
\]
Section~\ref{sec:local-m5} evaluates this expression for the standard \(SU(2)\)
representative of the M5 magnetic flux.

\subsection{What is deliberately not used}

Several other mechanisms can produce odd \(K\)-classes in noncommutative geometry:
crossed-product return unitaries, Connes--Thom classes, longitudinal Dirac operators,
or general odd Kasparov cycles. Those mechanisms are important in their proper settings.
They are not the mechanism used for the basic M5 source charge in this paper.

In particular, for the local M5-brane source there is no natural first-return map selecting
a crossed-product implementing unitary. There is also no need to choose a longitudinal
Dirac class as the primary parity-fixing object. The physical data first give a linking
flux class
\[
  x_{\check C}\in K_0(A_{\link}),
\]
and the source extension then gives
\[
  [u_{\check C}]=\partial x_{\check C}.
\]
This avoids an ad hoc choice of an odd generator and keeps the construction tied to the
physical linking charge.

Similarly, the class
\[
  x_{\check C}
\]
should not be confused with the even transverse-geometric \(K\)-class used in the
published even-codimensional NCFI. In the present paper the \(K\)-class is a charge
class associated with the \(C\)-field flux on the link. In geometric test examples, one may
pair the same cyclic cocycle with another \(K\)-class, such as a tangent orbibundle class,
but that produces a geometric characteristic number, not automatically an M5 magnetic
charge.

This distinction will be used throughout the paper. The local \(S^4\) computation uses a
flux class. The noncommutative-torus compactification uses a compactified charge class.
The orbifold calculation later in the paper is explicitly labelled as a geometric test
pairing unless its \(K\)-class is identified with a shifted \(C\)-field flux class.

\section{The local M5 linking theorem}
\label{sec:local-m5}

We now prove the local normalization result. This is the first genuinely physical
calculation in the paper. No noncommutative torus, orbifold groupoid or singular
transverse quotient is used in this section. The link is the ordinary four-sphere around
a smooth M5-brane source, and the source charge is the usual magnetic flux integral.

Let
\[
  \Wsix\hookrightarrow \Xeleven
\]
be an M5-brane worldvolume. Locally near a point of \(W\), the normal geometry is
\[
  \bbR^5,
\]
and the source lies at the origin of the normal fibre. A small normal ball is
\[
  D^5,
\]
and its boundary is the linking sphere
\[
  \partial D^5=S^4_{\link}.
\]
The punctured normal fibre is
\[
  \bbR^5\setminus\{0\}\simeq (0,\infty)\times S^4.
\]
Thus the magnetic charge of the source is measured on the four-dimensional link
\[
  S^4_{\link}.
\]

The M-theory \(C\)-field has field strength \(G_4\). In the presence of an M5-brane
source, the Bianchi identity is modified by a magnetic current. Locally, for a source of
charge \(N\), one writes schematically
\[
  \dd G_4=2\pi N\,\delta_W.
\]
The charge measured on the link is
\[
  Q_{\Mfive}
  =
  \int_{S^4_{\link}}\frac{G_4}{2\pi}.
\]
As recalled in the introduction, the \(C\)-field is shifted and globally refined in quantum
M-theory \cite{WittenFlux,DiaconescuFreedMoore,HopkinsSinger}. In the present section
we work at the de Rham, or Chern-character, level. This is the level seen by the
Chern--Connes pairing.

\subsection{The fundamental cyclic cocycle of the link}

Let
\[
  A_{\link}^{\infty}=C^\infty(S^4),
  \qquad
  A_{\link}=C(S^4).
\]
Fix an orientation of \(S^4\). Let
\[
  \tau_{S^4}\in HC^4(C^\infty(S^4))
\]
be the fundamental cyclic four-cocycle, normalized so that its Chern--Connes pairing
with a smooth vector bundle \(E\to S^4\) is
\[
  \pair{\tau_{S^4}}{[E]}
  =
  \int_{S^4}\ch(E)_{[4]}.
\]
Equivalently, it may be represented by
\[
  \tau_{S^4}(a^0,\ldots,a^4)
  =
  c_4
  \int_{S^4}
  a^0\,\dd a^1\wedge \dd a^2\wedge \dd a^3\wedge \dd a^4,
\]
for
\[
  a^0,\ldots,a^4\in C^\infty(S^4),
\]
where the constant \(c_4\) is fixed by the normalization above. This convention absorbs
the standard factorial and \(2\pi i\)-factors in the even Chern--Connes character.

The source charge class will be a class
\[
  x_{\check C}\in K_0(C(S^4))
\]
whose Chern character has degree-four component
\[
  \ch(x_{\check C})_{[4]}
  =
  \left[\frac{G_4}{2\pi}\right]_{S^4}
  \in H^4_{\mathrm{dR}}(S^4).
\]
The local source charge is then written as
\[
  Q_{\src,\Mfive}^{\mathrm{loc}}
  :=
  \pair{\tau_{S^4}}{x_{\check C}}.
\]
The theorem below verifies that this formula reproduces the ordinary M5 magnetic charge
in the standard integral model.

\subsection{The \texorpdfstring{\(SU(2)\)}{SU(2)} representative of the flux class}

Complex rank-two \(SU(2)\)-bundles over \(S^4\) are classified by their second Chern
class, equivalently by the integer
\[
  \int_{S^4}c_2(E).
\]
This is the standard classification of \(SU(2)\)-bundles over the four-sphere, using
\[
  \pi_3(SU(2))\simeq \bbZ
\]
and the Chern-Weil normalization of \(c_2\); see \cite{MilnorStasheff,BottTu}.

Let
\[
  E_N\longrightarrow S^4
\]
be a rank-two \(SU(2)\)-bundle satisfying
\[
  \int_{S^4}c_2(E_N)=N,
  \qquad
  N\in \bbZ.
\]
Since \(E_N\) has structure group \(SU(2)\),
\[
  c_1(E_N)=0.
\]
The degree-four component of its Chern character is therefore
\[
  \ch(E_N)_{[4]}
  =
  \ch_2(E_N)
  =
  \frac12\bigl(c_1(E_N)^2-2c_2(E_N)\bigr)
  =
  -c_2(E_N).
\]
Thus the virtual class
\[
  [E_N]-2[1]
\]
has degree-four Chern character \(-c_2(E_N)\). With the orientation convention chosen
above, the positive charge \(N\) is represented by the opposite virtual class
\[
  x_N:=2[1]-[E_N]\in \widetilde K^0(S^4)\subset K_0(C(S^4)).
\]
Then
\[
  \ch(x_N)_{[4]}=c_2(E_N).
\]

\begin{theorem}
\label{thm:local-linking-charge}
Let \(S^4\) be oriented and let
\[
  \tau_{S^4}\in HC^4(C^\infty(S^4))
\]
be normalized by
\[
  \pair{\tau_{S^4}}{[E]}
  =
  \int_{S^4}\ch(E)_{[4]}.
\]
Let \(E_N\to S^4\) be a rank-two \(SU(2)\)-bundle satisfying
\[
  \int_{S^4}c_2(E_N)=N.
\]
Define
\[
  x_N:=2[1]-[E_N]\in K_0(C(S^4)).
\]
Then
\[
  \pair{\tau_{S^4}}{x_N}=N.
\]
If the local \(C\)-field flux on the linking sphere satisfies
\[
  \left[\frac{G_4}{2\pi}\right]_{S^4}
  =
  c_2(E_N),
\]
then
\[
  \pair{\tau_{S^4}}{x_N}
  =
  \int_{S^4}\frac{G_4}{2\pi}
  =
  Q_{\Mfive}.
\]
\end{theorem}

\begin{proof}
By the normalization of \(\tau_{S^4}\),
\[
  \pair{\tau_{S^4}}{x_N}
  =
  \int_{S^4}\ch(x_N)_{[4]}.
\]
Since
\[
  x_N=2[1]-[E_N],
\]
and the trivial class \(2[1]\) has no degree-four Chern-character component, we have
\[
  \ch(x_N)_{[4]}
  =
  -\ch(E_N)_{[4]}.
\]
For an \(SU(2)\)-bundle,
\[
  c_1(E_N)=0,
\]
so
\[
  \ch(E_N)_{[4]}
  =
  \frac12(0-2c_2(E_N))
  =
  -c_2(E_N).
\]
Hence
\[
  \ch(x_N)_{[4]}=c_2(E_N).
\]
Therefore
\[
  \pair{\tau_{S^4}}{x_N}
  =
  \int_{S^4}c_2(E_N)
  =
  N.
\]
If
\[
  \left[\frac{G_4}{2\pi}\right]_{S^4}=c_2(E_N),
\]
then
\[
  \int_{S^4}c_2(E_N)
  =
  \int_{S^4}\frac{G_4}{2\pi}.
\]
This gives the claimed equality with the ordinary M5 magnetic charge.
\end{proof}

\begin{remark}
The theorem does not assert that the M-theory \(C\)-field is literally an \(SU(2)\)-bundle.
The \(SU(2)\)-bundle is a convenient \(K\)-theoretic representative of the integral
degree-four class on the linking sphere. Globally, the \(C\)-field is shifted and refined.
The Chern--Connes pairing used here sees only its Chern-character-level image on the
link.
\end{remark}

\subsection{The boundary source class}

We now express the same local charge in the odd codimension-five form described in
Section~\ref{sec:published-ncfi-to-source}. The local source extension is
\[
  0
  \longrightarrow
  C_0(\bbR^5)
  \longrightarrow
  C(D^5)
  \longrightarrow
  C(S^4)
  \longrightarrow
  0.
\]
Here
\[
  C_0(\bbR^5)\simeq C_0(\mathring D^5)
\]
is the ideal of functions vanishing on the boundary, and the quotient is restriction to
\[
  \partial D^5=S^4.
\]
The six-term exact sequence gives a boundary map
\[
  \partial:
  K_0(C(S^4))
  \longrightarrow
  K_1(C_0(\bbR^5)).
\]
For the flux representative \(x_N\), define
\[
  [u_N]:=\partial x_N\in K_1(C_0(\bbR^5)).
\]
More generally, for a local \(C\)-field charge class
\[
  x_{\check C}\in K_0(C(S^4)),
\]
define the associated odd source class by
\[
  [u_{\check C}]
  :=
  \partial x_{\check C}
  \in K_1(C_0(\bbR^5)).
\]
This is the codimension-five odd \(K\)-class. It is not an independent choice; it is the
boundary image of the physical linking class.

Let
\[
  \phi_{\bbR^5}\in HC^5(C_c^\infty(\bbR^5))
\]
be the odd cyclic cocycle on the punctured normal fibre, normalized so that its boundary
residue is the fundamental cyclic cocycle of the oriented linking sphere:
\[
  \Res(\phi_{\bbR^5})=\tau_{S^4}.
\]
Equivalently, \(\phi_{\bbR^5}\) is the odd Thom-type cyclic cocycle compatible with the
Bott boundary map for the extension above. With the chosen orientation convention, the
boundary compatibility identity is
\[
  \pair{\phi_{\bbR^5}}{\partial x}
  =
  \pair{\tau_{S^4}}{x}
\]
for the smooth \(K\)-classes under consideration.

\begin{proposition}
\label{prop:local-boundary-form}
With the orientation convention fixed above,
\[
  \pair{\phi_{\bbR^5}}{\partial x_N}
  =
  \pair{\tau_{S^4}}{x_N}
  =
  N.
\]
If
\[
  \left[\frac{G_4}{2\pi}\right]_{S^4}=c_2(E_N),
\]
then
\[
  \pair{\phi_{\bbR^5}}{[u_N]}
  =
  \int_{S^4}\frac{G_4}{2\pi}.
\]
\end{proposition}

\begin{proof}
The source extension
\[
  0
  \longrightarrow
  C_0(\bbR^5)
  \longrightarrow
  C(D^5)
  \longrightarrow
  C(S^4)
  \longrightarrow
  0
\]
is the local Bott boundary extension. Its \(K\)-theory boundary map sends the even
linking class to the odd class on the punctured normal fibre:
\[
  [u_N]=\partial x_N.
\]
By the chosen normalization of the odd cocycle \(\phi_{\bbR^5}\), its cyclic residue is
\[
  \Res(\phi_{\bbR^5})=\tau_{S^4},
\]
and the Chern--Connes pairing is compatible with the boundary map:
\[
  \pair{\phi_{\bbR^5}}{\partial x_N}
  =
  \pair{\tau_{S^4}}{x_N}.
\]
Theorem~\ref{thm:local-linking-charge} gives
\[
  \pair{\tau_{S^4}}{x_N}=N.
\]
If the flux class satisfies
\[
  [G_4/2\pi]_{S^4}=c_2(E_N),
\]
then the same theorem identifies this value with
\[
  \int_{S^4}\frac{G_4}{2\pi}.
\]
\end{proof}

\begin{remark}
The even formula
\[
  \pair{\tau_{S^4}}{x_N}
\]
is the physically transparent linking-charge formula. The odd formula
\[
  \pair{\phi_{\bbR^5}}{\partial x_N}
\]
is the codimension-five boundary form of the same charge. The boundary map is the
mechanism which supplies the required odd \(K\)-class.
\end{remark}

\subsection{Local source charge notation}

For later comparison with compactified and quotient examples, we record the local
notation. A local smooth M5 source datum is
\[
  \mathfrak D_{\Mfive}^{\mathrm{loc}}
  =
  \bigl(
    S^4_{\link},
    C^\infty(S^4),
    \tau_{S^4},
    x_{\check C}
  \bigr),
\]
where
\[
  x_{\check C}\in K_0(C(S^4))
\]
satisfies
\[
  \ch(x_{\check C})_{[4]}
  =
  \left[\frac{G_4}{2\pi}\right]_{S^4}.
\]
The local source charge is
\[
  Q_{\src,\Mfive}^{\mathrm{loc}}
  :=
  \pair{\tau_{S^4}}{x_{\check C}}.
\]
The associated odd boundary form is
\[
  Z_{\mathrm{bdry},\Mfive}^{\mathrm{loc}}
  :=
  \pair{\phi_{\bbR^5}}{\partial x_{\check C}}.
\]
When the boundary compatibility convention is fixed as above,
\[
  Z_{\mathrm{bdry},\Mfive}^{\mathrm{loc}}
  =
  Q_{\src,\Mfive}^{\mathrm{loc}}.
\]

In the specific representative
\[
  x_{\check C}=x_N=2[1]-[E_N],
\]
one obtains
\[
  Q_{\src,\Mfive}^{\mathrm{loc}}
  =
  Z_{\mathrm{bdry},\Mfive}^{\mathrm{loc}}
  =
  N.
\]
This is the local normalization which every noncommutative, compactified or groupoid
extension must preserve in its commutative limit.

\section{Connes--Douglas--Schwarz compactified source sectors}
\label{sec:cds}

We now pass from the smooth local link \(S^4\) to a compactified source sector of
Connes--Douglas--Schwarz type. The purpose of this section is not to replace the local
M5-brane linking sphere by a noncommutative torus in all situations. Rather, the point is
more specific: when the physical compactification sector is described by a
noncommutative torus, the same source-charge principle becomes a Chern--Connes
pairing over the corresponding noncommutative torus algebra.

The physical input is the Matrix/M-theory compactification mechanism. In the BFSS
proposal, M-theory is described, in an appropriate infinite-momentum or DLCQ regime,
by a large-\(N\) matrix model \cite{BanksFischlerShenkerSusskind}. Connes, Douglas
and Schwarz showed that toroidal compactifications of Matrix theory naturally lead to
noncommutative tori, and argued that the deformation parameter is related to background
three-form data in the corresponding supergravity description \cite{ConnesDouglasSchwarz}.
Thus the noncommutative torus algebra is a physical compactification algebra, not an
arbitrary algebraic replacement.

This is the strongest noncommutative-geometry bridge in the paper. It connects the
source-charge pairing to an already established appearance of noncommutative geometry
in Matrix/M-theory. The role of this section is therefore to define the compactified
source charge
\[
  Q_{\src,\Theta}
  =
  \pair{\tau_\Theta^{(4)}}{x_{\check C,\Theta}},
\]
to check its commutative limit, and to impose Morita invariance as the appropriate
duality-compatibility condition.

\subsection{Smooth noncommutative four-tori}

Let
\[
  \Theta=(\Theta_{ij})_{1\leq i,j\leq 4}
\]
be a real skew-symmetric \(4\times 4\) matrix. The \(C^*\)-algebra of the
four-dimensional noncommutative torus is the universal \(C^*\)-algebra
\[
  A_\Theta=C(\bbT^4_\Theta)
\]
generated by unitaries
\[
  U_1,U_2,U_3,U_4
\]
satisfying
\[
  U_iU_j
  =
  e^{2\pi i\Theta_{ij}}U_jU_i,
  \qquad
  \Theta_{ij}=-\Theta_{ji}.
\]
The smooth noncommutative torus is the Fréchet algebra
\[
  A_\Theta^\infty
  =
  \left\{
  \sum_{n\in\bbZ^4}a_nU^n:
  (a_n)\in\mathcal S(\bbZ^4)
  \right\},
\]
where
\[
  U^n:=U_1^{n_1}U_2^{n_2}U_3^{n_3}U_4^{n_4}.
\]
This is the algebra of smooth vectors for the canonical gauge action of \(\bbT^4\) on
\(A_\Theta\). The basic analytic and \(K\)-theoretic structure of noncommutative tori
was developed by Rieffel, Elliott, Connes--Rieffel, Schwarz and Rieffel--Schwarz
\cite{RieffelIrrational,ElliottKTheoryTori,RieffelProjectiveModules,
ConnesRieffelYangMills,RieffelTori,SchwarzMorita,RieffelSchwarz}.

The canonical trace is
\[
  \tau_0\!\left(\sum_{n\in\bbZ^4}a_nU^n\right)=a_0.
\]
The standard derivations are
\[
  \delta_j(U_k)=2\pi i\,\delta_{jk}U_k,
  \qquad
  1\leq j,k\leq4.
\]
They generate the differential calculus used in the cyclic cocycles below.

The fundamental cyclic four-cocycle is represented by
\[
  \tau_\Theta^{(4)}(a^0,\ldots,a^4)
  =
  c_\Theta
  \sum_{\sigma\in S_4}\operatorname{sgn}(\sigma)\,
  \tau_0\!\left(
    a^0
    \delta_{\sigma(1)}(a^1)
    \delta_{\sigma(2)}(a^2)
    \delta_{\sigma(3)}(a^3)
    \delta_{\sigma(4)}(a^4)
  \right),
\]
where
\[
  a^0,\ldots,a^4\in A_\Theta^\infty.
\]
The constant \(c_\Theta\) is fixed by the convention that, when \(\Theta=0\), this
cocycle agrees with the ordinary fundamental cyclic cocycle of the commutative torus
\(\bbT^4\):
\[
  \tau_{\bbT^4}(f^0,\ldots,f^4)
  =
  c_{\bbT^4}
  \int_{\bbT^4}
  f^0\,\dd f^1\wedge\dd f^2\wedge\dd f^3\wedge\dd f^4.
\]
With this normalization, the pairing with a \(K_0\)-class gives the corresponding
noncommutative four-dimensional Chern number.

\begin{definition}
\label{def:cds-compactified-datum}
A \emph{CDS-type compactified source datum} consists of:
\begin{enumerate}[label=(\roman*)]
  \item a smooth noncommutative four-torus
  \[
    A_\Theta^\infty\subset A_\Theta;
  \]

  \item the fundamental cyclic four-cocycle
  \[
    \tau_\Theta^{(4)}\in HC^4(A_\Theta^\infty);
  \]

  \item a charge class
  \[
    x_{\check C,\Theta}\in K_0(A_\Theta)
  \]
  represented by a smooth projection or idempotent over \(A_\Theta^\infty\), possibly
  after stabilization.
\end{enumerate}
The associated compactified source charge is
\[
  Q_{\src,\Theta}
  :=
  \pair{\tau_\Theta^{(4)}}{x_{\check C,\Theta}}.
\]
\end{definition}

The notation \(x_{\check C,\Theta}\) is intentional. The class is meant to represent the
compactified \(C\)-field charge sector at the level of the Chern character. It is not the
full shifted differential \(C\)-field. The refined flux quantization problem is postponed to
Section~\ref{sec:cfield-anomalies}.

\subsection{Commutative limit}

The compactified source charge must reduce to an ordinary characteristic-number
pairing when the noncommutativity is turned off. Set
\[
  \Theta=0.
\]
Then
\[
  A_0=C(\bbT^4),
  \qquad
  A_0^\infty=C^\infty(\bbT^4),
\]
and \(\tau_\Theta^{(4)}\) becomes the ordinary fundamental cyclic four-cocycle of
\(\bbT^4\). If \(E\to\bbT^4\) is a smooth complex vector bundle and
\[
  x_0=[E]\in K^0(\bbT^4)=K_0(C(\bbT^4)),
\]
then
\[
  \pair{\tau_0^{(4)}}{x_0}
  =
  \int_{\bbT^4}\ch(E)_{[4]}.
\]

\begin{proposition}
\label{prop:cds-commutative-limit}
Let
\[
  x_0=[E]\in K_0(C(\bbT^4))
\]
be represented by a smooth complex vector bundle \(E\to \bbT^4\). With the
normalization above,
\[
  Q_{\src,0}(x_0)
  =
  \pair{\tau_0^{(4)}}{x_0}
  =
  \int_{\bbT^4}\ch(E)_{[4]}.
\]
\end{proposition}

\begin{proof}
When \(\Theta=0\), the smooth noncommutative torus is the ordinary commutative
algebra \(C^\infty(\bbT^4)\), and the derivations \(\delta_j\) are the usual coordinate
derivations. The cyclic cocycle \(\tau_0^{(4)}\) is normalized to be the fundamental
cyclic cocycle of \(\bbT^4\). The Chern--Connes pairing with a vector bundle class is
therefore the ordinary pairing of the degree-four Chern character with the fundamental
class of \(\bbT^4\).
\end{proof}

This proposition is the toroidal analogue of the local \(S^4\) normalization theorem. In
the local M5 model the four-dimensional cycle is a linking sphere. In the CDS compactified
sector, the four-dimensional cycle is a toroidal compactification or charge sector. The
same pairing principle applies in both cases.

\subsection{The compactified odd boundary class}

The pairing
\[
  \pair{\tau_\Theta^{(4)}}{x_{\check C,\Theta}}
\]
is even: a cyclic four-cocycle paired with a \(K_0\)-class. To express the same charge in
the strict odd codimension-five form, use the cone or suspension extension
\[
  0
  \longrightarrow
  SA_\Theta
  \longrightarrow
  CA_\Theta
  \longrightarrow
  A_\Theta
  \longrightarrow
  0,
\]
where
\[
  SA_\Theta=C_0((0,1))\otimes A_\Theta.
\]
The associated boundary map is
\[
  \partial:K_0(A_\Theta)\longrightarrow K_1(SA_\Theta).
\]
For the compactified charge class \(x_{\check C,\Theta}\), define
\[
  [u_{\check C,\Theta}]
  :=
  \partial x_{\check C,\Theta}
  \in K_1(SA_\Theta).
\]
This is the compactified analogue of the local class
\[
  [u_{\check C}]
  =
  \partial x_{\check C}
  \in K_1(C_0(\bbR^5)).
\]

Let
\[
  \phi_\Theta^{(5)}\in HC^5((SA_\Theta)^\infty)
\]
be an odd cyclic cocycle whose boundary residue or suspension transgression is
\[
  \Res(\phi_\Theta^{(5)})=\tau_\Theta^{(4)}.
\]
Then the strict odd boundary expression is
\[
  Z_{\mathrm{bdry},\Theta}
  :=
  \pair{\phi_\Theta^{(5)}}{\partial x_{\check C,\Theta}}.
\]
When the cyclic residue and \(K\)-theory boundary map are compatible,
\[
  \pair{\phi_\Theta^{(5)}}{\partial x_{\check C,\Theta}}
  =
  \varepsilon
  \pair{\tau_\Theta^{(4)}}{x_{\check C,\Theta}},
  \qquad
  \varepsilon\in\{+1,-1\}.
\]
The sign is fixed by the orientation convention for the cone boundary.

\begin{proposition}
\label{prop:cds-boundary-form}
For a boundary-compatible CDS-type compactified source datum,
\[
  Z_{\mathrm{bdry},\Theta}
  =
  \varepsilon\,Q_{\src,\Theta}.
\]
In particular, after fixing the boundary orientation so that \(\varepsilon=+1\), the strict
odd boundary expression equals the even compactified source pairing.
\end{proposition}

\begin{proof}
By definition,
\[
  Z_{\mathrm{bdry},\Theta}
  =
  \pair{\phi_\Theta^{(5)}}{\partial x_{\check C,\Theta}}.
\]
The assumed boundary compatibility gives
\[
  \pair{\phi_\Theta^{(5)}}{\partial x_{\check C,\Theta}}
  =
  \varepsilon
  \pair{\tau_\Theta^{(4)}}{x_{\check C,\Theta}}.
\]
The last pairing is \(Q_{\src,\Theta}\). This proves the claim.
\end{proof}

This is the same parity mechanism as in the local \(S^4\) case. The physical charge is
first represented by an even linking or compactification class. The odd class is then the
boundary image of that class.

\subsection{Morita equivalence and duality compatibility}

The CDS compactification picture requires Morita invariance. Morita-equivalent
noncommutative tori correspond to equivalent compactification data under the relevant
duality group. In dimension \(d\), the Morita equivalence classes are controlled by an
integral \(O(d,d;\bbZ)\)-type action on the deformation parameter, as described by
Schwarz and by Rieffel--Schwarz \cite{SchwarzMorita,RieffelSchwarz}.

Let
\[
  A_\Theta\sim_M A_{\Theta'}
\]
be a Morita equivalence of noncommutative four-tori, implemented by an equivalence
bimodule
\[
  \calE_{\Theta,\Theta'}.
\]
The bimodule transports \(K\)-theory classes:
\[
  x_{\check C,\Theta'}
  =
  (\calE_{\Theta,\Theta'})_*(x_{\check C,\Theta})
  \in K_0(A_{\Theta'}).
\]
At the smooth level, the same Morita equivalence transports the cyclic cocycle:
\[
  \tau_{\Theta'}^{(4)}
  =
  (\calE_{\Theta,\Theta'})_*(\tau_\Theta^{(4)}).
\]
The Chern--Connes pairing is invariant under this simultaneous transport.

\begin{proposition}
\label{prop:cds-morita-invariance}
Let
\[
  A_\Theta\sim_M A_{\Theta'}
\]
be Morita-equivalent noncommutative four-tori. Suppose that the charge class and the
fundamental cyclic cocycle are transported through the same Morita equivalence. Then
\[
  Q_{\src,\Theta}
  =
  Q_{\src,\Theta'}.
\]
Equivalently,
\[
  \pair{\tau_\Theta^{(4)}}{x_{\check C,\Theta}}
  =
  \pair{\tau_{\Theta'}^{(4)}}{x_{\check C,\Theta'}}.
\]
\end{proposition}

\begin{proof}
This is the Morita invariance of the Chern--Connes pairing. The equivalence bimodule
identifies the \(K\)-theory groups and transports cyclic cocycles on the corresponding
smooth noncommutative tori. Naturality of the Chern character under Morita equivalence
gives equality of the numerical pairings.
\end{proof}

This proposition is the precise compatibility condition with CDS compactifications. The
charge is not tied to a particular presentation of the noncommutative torus. It is tied to
the transported pair
\[
  \bigl(\tau_\Theta^{(4)},x_{\check C,\Theta}\bigr)
\]
inside the Morita class.

\subsection{Twisted group-algebra viewpoint}

It is sometimes useful to view the noncommutative torus as a twisted group algebra. Let
\[
  \sigma_\Theta:\bbZ^4\times\bbZ^4\longrightarrow U(1)
\]
be the multiplier
\[
  \sigma_\Theta(m,n)
  =
  \exp\!\left(\pi i\,m^T\Theta n\right),
  \qquad
  m,n\in\bbZ^4.
\]
Then \(A_\Theta\) is isomorphic to the twisted group \(C^*\)-algebra
\[
  C^*(\bbZ^4,\sigma_\Theta),
\]
up to the conventional choice of normalization of the cocycle. This is the algebraic
form most directly related to projective representations of the compactification lattice.
For twisted crossed products and noncommutative torus models, see
\cite{PackerRaeburn,RieffelIrrational,ElliottKTheoryTori,RieffelProjectiveModules}.

In this picture, the deformation parameter \(\Theta\) records the projective phase of
lattice translations. The cyclic cocycle \(\tau_\Theta^{(4)}\) is the noncommutative
fundamental class of this compactified four-dimensional sector, and the charge class
\[
  x_{\check C,\Theta}\in K_0(C^*(\bbZ^4,\sigma_\Theta))
\]
is the projective-module class representing the compactified flux sector. The pairing
\[
  \pair{\tau_\Theta^{(4)}}{x_{\check C,\Theta}}
\]
is therefore the noncommutative Chern number of the corresponding projective module.

\subsection{Product compactifications}

There is also a product version which keeps the local M5 linking sphere and adds an
internal noncommutative compactification sector. Let
\[
  A_{\mathrm{loc}}^\infty=C^\infty(S^4),
  \qquad
  A_{\mathrm{int}}^\infty=A_\Theta^\infty(\bbT^{2r}_\Theta),
\]
and define
\[
  A_{\mathrm{tot}}^\infty
  =
  C^\infty(S^4)\,\widehat\otimes\,A_\Theta^\infty(\bbT^{2r}_\Theta).
\]
Let
\[
  \tau_{S^4}\in HC^4(C^\infty(S^4))
\]
be the local linking cocycle, and let
\[
  \tau_\Theta^{(2r)}\in HC^{2r}(A_\Theta^\infty(\bbT^{2r}_\Theta))
\]
be the fundamental cyclic cocycle of the internal noncommutative torus. Their external
product is
\[
  \tau_{\mathrm{tot}}
  :=
  \tau_{S^4}\#\tau_\Theta^{(2r)}
  \in HC^{4+2r}(A_{\mathrm{tot}}^\infty).
\]
For
\[
  x_N\in K_0(C(S^4)),
  \qquad
  y_\Theta\in K_0(A_\Theta(\bbT^{2r}_\Theta)),
\]
define
\[
  x_{\mathrm{tot}}:=x_N\boxtimes y_\Theta.
\]
The Chern--Connes pairing factorizes:
\[
  \pair{\tau_{S^4}\#\tau_\Theta^{(2r)}}{x_N\boxtimes y_\Theta}
  =
  \pair{\tau_{S^4}}{x_N}\,
  \pair{\tau_\Theta^{(2r)}}{y_\Theta}.
\]

\begin{proposition}
\label{prop:product-compactified-charge}
For the product of a local \(S^4\) M5 linking sector with an internal noncommutative
torus sector,
\[
  \pair{\tau_{S^4}\#\tau_\Theta^{(2r)}}{x_N\boxtimes y_\Theta}
  =
  \pair{\tau_{S^4}}{x_N}\,
  \pair{\tau_\Theta^{(2r)}}{y_\Theta}.
\]
If
\[
  \pair{\tau_{S^4}}{x_N}=N,
\]
then
\[
  \pair{\tau_{\mathrm{tot}}}{x_{\mathrm{tot}}}
  =
  N\,\pair{\tau_\Theta^{(2r)}}{y_\Theta}.
\]
\end{proposition}

\begin{proof}
This is the multiplicativity of the Chern character under external products in
\(K\)-theory and cyclic cohomology. The exterior product of cyclic cocycles pairs with
the exterior product of \(K\)-classes as the product of the individual pairings.
\end{proof}

The product number is a compactification-dressed charge. It should not be confused with
the bare local M5 magnetic charge. The local charge is
\[
  \pair{\tau_{S^4}}{x_N}.
\]
The internal noncommutative torus contributes the additional topological factor
\[
  \pair{\tau_\Theta^{(2r)}}{y_\Theta}.
\]

\subsection{Relation to \texorpdfstring{\(T\)}{T}-duality and flux backgrounds}

Noncommutative tori also appear naturally in the topology of \(T\)-duality with flux.
D-branes in \(B\)-field backgrounds are classified by twisted \(K\)-theory, and torus
bundles with \(H\)-flux can have noncommutative \(T\)-duals
\cite{BouwknegtMathaiTwistedK,BouwknegtEvslinMathai,MathaiRosenbergTDuality}. This
is not identical to the CDS Matrix-theory compactification mechanism, but it reinforces
the same structural lesson: background flux can force the effective charge algebra to be
noncommutative.

For the present paper, the relevant conclusion is modest. The source charge is computed
by a Chern--Connes pairing whenever the compactification or linking sector is physically
represented by a noncommutative torus algebra. The pairing sees the Chern-character
part of the charge. The integral and torsion refinements belong to the full \(C\)-field or
twisted \(K\)-theoretic quantization problem.

\subsection{What is not being claimed}

The CDS connection should not be overstated. The deformation parameter
\[
  \Theta
\]
is compactification or background data. It is not itself the M5 magnetic charge. The
charge is the value of the pairing
\[
  \pair{\tau_\Theta^{(4)}}{x_{\check C,\Theta}}.
\]
Similarly, the class
\[
  x_{\check C,\Theta}\in K_0(A_\Theta)
\]
is not the full shifted differential \(C\)-field. It is a \(K\)-theoretic representative of the
Chern-character-level charge sector.

Thus the claim of this section is exactly this:
\begin{center}
\fbox{%
\begin{minipage}{0.86\textwidth}
\centering
In a CDS-type compactified sector, the M5 source charge admits a Chern--Connes
expression over a noncommutative torus algebra, with Morita invariance as its
duality-compatibility condition.
\end{minipage}}
\end{center}
This is a compactification-compatible extension of the local \(S^4\) source theorem, not
a derivation of Matrix theory or of the full M5-brane dynamics.

\section{Orbifold source links}
\label{sec:orbifold-link}

We now consider quotient linking geometries. The purpose of this section is not to claim
that every M5-brane source has an orbifold link. The ordinary local link of a smooth
embedded M5-brane in a smooth eleven-dimensional spacetime is \(S^4\). Orbifold links
enter only when the physical transverse geometry has quotient singularities or finite
stabilizers. In that case the correct linking object is not the naive quotient space, but the
orbifold or groupoid which retains isotropy data.

This is standard in string/M-theory. Fivebranes and M-theory on orbifolds already occur
in \cite{WittenOrbifoldFivebranes}. D-branes at orbifold singularities produce quiver
gauge theories \cite{DouglasMooreOrbifolds}. Later, ABJM theory gave a central example
of M2-branes probing \(\bbC^4/\bbZ_k\), with an M-theory dual involving
\(S^7/\bbZ_k\) \cite{ABJM}. These examples justify the use of quotient and orbifold
geometries in brane physics. The groupoid formalism is used here because it retains the
stabilizer data which the naive quotient space forgets.

The mathematical background on orbifolds starts with Satake's \(V\)-manifolds and
Kawasaki's index theory \cite{SatakeVManifolds,SatakeGaussBonnet,KawasakiIndex,
KawasakiRRVManifolds}. The groupoid approach to orbifolds is standard
\cite{MoerdijkPronk,AdemLeidaRuan,TuXuLaurentGengoux}. For weighted projective
spaces and their characteristic classes we use \cite{KawasakiWeighted,DolgachevWeighted,
IanoFletcherWeighted,BahriFranzRay,BorisovChenSmith}. For related intersection-theoretic
and stack-theoretic background, see \cite{EdidinGrahamEquivariant,VistoliIntersectionStacks}.
Twisted and stringy orbifold theories are treated in \cite{ChenRuanOrbifoldCohomology,
LupercioUribeGerbes}; these refinements are not used in the basic untwisted computation
below.

\subsection{Orbifold links as physically admissible linking sectors}

Let \(L\) be a compact oriented real four-dimensional orbifold, presented by a proper Lie
groupoid
\[
  G_L\rightrightarrows G_L^{(0)}.
\]
The associated reduced groupoid \(C^*\)-algebra is
\[
  A_L:=C_r^*(G_L),
\]
and a smooth convolution model is denoted
\[
  A_L^\infty:=C_c^\infty(G_L,\Omega^{1/2}),
\]
or any fixed smooth Morita-equivalent orbifold atlas.

The fundamental cyclic cocycle of the oriented orbifold link is written
\[
  \tau_L\in HC^4(A_L^\infty).
\]
For an orbibundle \(E\to L\), represented by a \(K_0\)-class
\[
  [E]\in K_0(A_L),
\]
the Chern--Connes pairing with the untwisted orbifold fundamental cyclic cocycle is
\[
  \pair{\tau_L}{[E]}
  =
  \int_L^{\orb}\ch(E)_{[4]}.
\]
This is the orbifold analogue of the ordinary de Rham formula
\[
  \pair{\tau_{S^4}}{[E]}
  =
  \int_{S^4}\ch(E)_{[4]}.
\]

For an M5-brane source interpretation, the class paired with \(\tau_L\) should be a
charge class
\[
  x_{\check C}\in K_0(A_L)
\]
whose Chern--Connes character represents the restricted \(C\)-field flux at the
degree-four level:
\[
  \Ch(x_{\check C})_{[4]}
  =
  \left[\frac{G_4}{2\pi}\right]_L.
\]
Then the source charge on the orbifold link is
\[
  Q_{\src}(L;\check C)
  =
  \pair{\tau_L}{x_{\check C}}.
\]
If \(L\) is a genuine physical quotient link, this is the orbifold version of the local
linking charge.

\begin{remark}
The use of the groupoid \(G_L\) is not an optional complication. If \(L\) has stabilizer
groups, the naive topological quotient loses data. The action or orbifold groupoid
retains isotropy, and the Chern--Connes pairing is computed on the corresponding
groupoid algebra. Morita-equivalent orbifold atlases give the same pairing after
transport of the cyclic cocycle and the \(K\)-class.
\end{remark}

\subsection{Weighted projective links}

The concrete example used here is the weighted projective orbifold
\[
  X=\Porb(1,2,3).
\]
It is a compact complex orbifold surface and hence a compact oriented real
four-dimensional orbifold. It may be used as the link of the real five-dimensional cone
\[
  C(X)=([0,\epsilon)\times X)/(\{0\}\times X).
\]
Thus \(C(X)\) is a singular codimension-five transverse model, with \(X\) as its linking
orbifold. This is the orbifold analogue of the smooth local model
\[
  D^5,
  \qquad
  \partial D^5=S^4.
\]

The orbifold \(X=\Porb(1,2,3)\) is presented by the weighted Hopf action
\[
  S^1\curvearrowright S^5\subset\bbC^3,
\]
\[
  \lambda\cdot(z_0,z_1,z_2)
  =
  (\lambda z_0,\lambda^2z_1,\lambda^3z_2),
  \qquad
  \lambda\in S^1.
\]
The infinitesimal generator is
\[
  Y(z_0,z_1,z_2)=(iz_0,2iz_1,3iz_2).
\]
Since at least one coordinate is nonzero on \(S^5\), and all weights are positive, \(Y\)
is nowhere zero. The \(S^1\)-orbits define a regular one-dimensional foliation of \(S^5\).
The quotient is
\[
  X=S^5/S^1_{(1,2,3)}
  =
  \Porb(1,2,3).
\]

The action is locally free but not free. The stabilizers of the three coordinate orbits are
\[
  \operatorname{Stab}(1,0,0)=\{1\},
\]
\[
  \operatorname{Stab}(0,1,0)=\mu_2,
\]
and
\[
  \operatorname{Stab}(0,0,1)=\mu_3.
\]
Thus the quotient is not a smooth manifold. The orbifold stabilizers are finite and are
retained by the action groupoid
\[
  G_X:=S^5\rtimes S^1.
\]
An arrow is a pair \((z,\lambda)\), with source and range maps
\[
  s(z,\lambda)=z,
  \qquad
  r(z,\lambda)=\lambda\cdot z.
\]
The inverse and composition are
\[
  (z,\lambda)^{-1}=(\lambda\cdot z,\lambda^{-1}),
\]
and
\[
  (\lambda\cdot z,\mu)\circ(z,\lambda)=(z,\mu\lambda).
\]
Since \(S^1\) is compact, the action groupoid is proper. Since the weights
\((1,2,3)\) have greatest common divisor one and the local slice representations are
effective, this groupoid presents the effective orbifold \(X=\Porb(1,2,3)\).

The reduced groupoid \(C^*\)-algebra is
\[
  A_X:=C_r^*(G_X)\simeq C(S^5)\rtimes S^1.
\]
At the smooth level one may use the smooth crossed product
\[
  A_X^\infty:=C^\infty(S^5)\rtimes S^1,
\]
or a smooth proper orbifold atlas Morita equivalent to the action groupoid.

\subsection{The orbifold fundamental cyclic cocycle}

Let
\[
  \tau_X\in HC^4(A_X^\infty)
\]
be the untwisted orbifold fundamental cyclic cocycle. For a complex orbibundle
\(E\to X\), represented in \(K_0(A_X)\), the pairing is
\[
  \pair{\tau_X}{[E]}
  =
  \int_X^{\orb}\ch(E)_{[4]}.
\]
The word ``untwisted'' is important. Full stringy or Chen--Ruan orbifold cohomology
contains twisted-sector data. The pairing used here is the ordinary orbifold fundamental
current pairing. It is the correct object for the basic Chern-character computation below.

If the orbifold link is used as a physical M5 source link, the charge class should be
\[
  x_{\check C}\in K_0(A_X),
\]
and the physical source charge would be
\[
  Q_{\src}(X;\check C)
  =
  \pair{\tau_X}{x_{\check C}}.
\]
The example computed below instead uses a geometric test class:
\[
  x_{\mathrm{geom}}
  :=
  [T_{\bbR}X\otimes\bbC]-4[1]
  \in K_0(A_X).
\]
The subtraction of \(4[1]\) removes the complex rank and does not affect the degree-four
Chern-character component.

Since \(X\) is a complex orbifold surface,
\[
  T_{\bbR}X\otimes\bbC
  \simeq
  T^{1,0}X\oplus T^{0,1}X.
\]
Let
\[
  E:=T^{1,0}X.
\]
For a complex rank-two bundle,
\[
  \ch_2(E)
  =
  \frac12\bigl(c_1(E)^2-2c_2(E)\bigr).
\]
The conjugate bundle \(\overline E=T^{0,1}X\) satisfies
\[
  c_1(\overline E)=-c_1(E),
  \qquad
  c_2(\overline E)=c_2(E),
\]
and hence
\[
  \ch_2(\overline E)=\ch_2(E).
\]
Therefore
\[
  \ch(T_{\bbR}X\otimes\bbC)_{[4]}
  =
  \ch_2(E)+\ch_2(\overline E)
  =
  c_1(E)^2-2c_2(E).
\]
This is the first Pontryagin class of the real tangent orbibundle:
\[
  \ch(T_{\bbR}X\otimes\bbC)_{[4]}
  =
  p_1(T_{\bbR}X).
\]
Consequently,
\[
  \pair{\tau_X}{x_{\mathrm{geom}}}
  =
  \int_X^{\orb}p_1(T_{\bbR}X).
\]

\subsection{Weighted projective characteristic classes}

We now compute the last integral. More generally, let
\[
  X=\Porb(a_0,a_1,a_2)
\]
be a weighted projective orbifold surface with positive pairwise coprime weights. Let
\[
  H:=c^{\orb}_1(\mathcal O_X(1))\in H^2(X;\bbQ)
\]
be the orbifold hyperplane class. The rational cohomology is generated by \(H\), with
\(H^3=0\), and the orbifold integration convention is
\[
  \int_{\Porb(a_0,a_1,a_2)}^{\orb}H^2
  =
  \frac{1}{a_0a_1a_2}.
\]
The weighted Euler sequence is
\[
  0
  \longrightarrow
  \mathcal O_X
  \longrightarrow
  \mathcal O_X(a_0)\oplus\mathcal O_X(a_1)\oplus\mathcal O_X(a_2)
  \longrightarrow
  T^{1,0}X
  \longrightarrow
  0.
\]
Hence
\[
  c(T^{1,0}X)
  =
  (1+a_0H)(1+a_1H)(1+a_2H).
\]
Since \(X\) has complex dimension two, we retain terms only up to \(H^2\). Thus
\[
  c_1(T^{1,0}X)
  =
  (a_0+a_1+a_2)H,
\]
and
\[
  c_2(T^{1,0}X)
  =
  (a_0a_1+a_0a_2+a_1a_2)H^2.
\]
Therefore
\[
  p_1(T_{\bbR}X)
  =
  c_1(T^{1,0}X)^2-2c_2(T^{1,0}X).
\]
Substituting the preceding formulae gives
\[
  p_1(T_{\bbR}X)
  =
  (a_0^2+a_1^2+a_2^2)H^2.
\]
Integrating,
\[
  \int_X^{\orb}p_1(T_{\bbR}X)
  =
  \frac{a_0^2+a_1^2+a_2^2}{a_0a_1a_2}.
\]

\begin{proposition}
\label{prop:weighted-projective-general}
Let
\[
  X=\Porb(a_0,a_1,a_2)
\]
be a weighted projective orbifold surface with positive pairwise coprime weights. Let
\[
  \tau_X\in HC^4(A_X^\infty)
\]
be the untwisted orbifold fundamental cyclic cocycle, and let
\[
  x_{\mathrm{geom}}
  =
  [T_{\bbR}X\otimes\bbC]-4[1]\in K_0(A_X).
\]
Then
\[
  \pair{\tau_X}{x_{\mathrm{geom}}}
  =
  \int_X^{\orb}p_1(T_{\bbR}X)
  =
  \frac{a_0^2+a_1^2+a_2^2}{a_0a_1a_2}.
\]
\end{proposition}

\begin{proof}
By the definition of the orbifold fundamental cyclic cocycle,
\[
  \pair{\tau_X}{x_{\mathrm{geom}}}
  =
  \int_X^{\orb}\ch(T_{\bbR}X\otimes\bbC)_{[4]}.
\]
Since \(X\) is a complex orbifold surface,
\[
  \ch(T_{\bbR}X\otimes\bbC)_{[4]}
  =
  p_1(T_{\bbR}X).
\]
The weighted Euler sequence gives
\[
  c_1(T^{1,0}X)=(a_0+a_1+a_2)H
\]
and
\[
  c_2(T^{1,0}X)=(a_0a_1+a_0a_2+a_1a_2)H^2.
\]
Therefore
\[
  p_1(T_{\bbR}X)
  =
  (a_0^2+a_1^2+a_2^2)H^2.
\]
Finally,
\[
  \int_X^{\orb}H^2
  =
  \frac{1}{a_0a_1a_2}.
\]
Combining these identities gives the formula.
\end{proof}

For the weights
\[
  (a_0,a_1,a_2)=(1,2,3),
\]
the formula gives
\[
  \pair{\tau_X}{x_{\mathrm{geom}}}
  =
  \frac{1^2+2^2+3^2}{1\cdot2\cdot3}
  =
  \frac{14}{6}
  =
  \frac73.
\]

\begin{corollary}
\label{cor:weighted-hopf-123}
Let
\[
  X=\Porb(1,2,3)
\]
be presented by the weighted Hopf action groupoid
\[
  G_X=S^5\rtimes S^1.
\]
Let
\[
  \tau_X\in HC^4(C^\infty(S^5)\rtimes S^1)
\]
be the untwisted orbifold fundamental cyclic cocycle, and let
\[
  x_{\mathrm{geom}}
  =
  [T_{\bbR}X\otimes\bbC]-4[1]\in K_0(C(S^5)\rtimes S^1).
\]
Then
\[
  \pair{\tau_X}{x_{\mathrm{geom}}}
  =
  \frac73.
\]
\end{corollary}

\begin{proof}
This is Proposition~\ref{prop:weighted-projective-general} with weights
\[
  (a_0,a_1,a_2)=(1,2,3).
\]
Indeed,
\[
  a_0^2+a_1^2+a_2^2=1+4+9=14,
\]
and
\[
  a_0a_1a_2=1\cdot2\cdot3=6.
\]
Thus
\[
  \pair{\tau_X}{x_{\mathrm{geom}}}
  =
  \frac{14}{6}
  =
  \frac73.
\]
\end{proof}

\subsection{Boundary form for the orbifold cone}

The preceding computation is an even pairing on the four-dimensional orbifold link.
To express it in the odd codimension-five source form, use the cone over the link:
\[
  C(X)=([0,\epsilon)\times X)/(\{0\}\times X).
\]
At the level of \(C^*\)-algebras one has, abstractly, a source or cone extension
\[
  0
  \longrightarrow
  A_{\bulk}
  \longrightarrow
  A_{\mathrm{src}}
  \longrightarrow
  A_X
  \longrightarrow
  0.
\]
The boundary map is
\[
  \partial:K_0(A_X)\longrightarrow K_1(A_{\bulk}).
\]
For the geometric test class define
\[
  [u_{\mathrm{geom}}]
  :=
  \partial x_{\mathrm{geom}}
  \in K_1(A_{\bulk}).
\]
Let
\[
  \phi_{\bulk}\in HC^5(A_{\bulk}^{\infty})
\]
be an odd bulk cyclic cocycle whose residue on the boundary link is
\[
  \Res(\phi_{\bulk})=\tau_X.
\]
If the boundary compatibility identity holds, then
\[
  \pair{\phi_{\bulk}}{\partial x_{\mathrm{geom}}}
  =
  \varepsilon\,\pair{\tau_X}{x_{\mathrm{geom}}},
  \qquad
  \varepsilon\in\{+1,-1\}.
\]
Thus, for \(X=\Porb(1,2,3)\),
\[
  \pair{\phi_{\bulk}}{[u_{\mathrm{geom}}]}
  =
  \varepsilon\,\frac73.
\]
The sign depends only on the boundary orientation convention.

\begin{proposition}
\label{prop:orbifold-boundary-pairing}
Let
\[
  X=\Porb(1,2,3)
\]
be used as a four-dimensional orbifold link. Suppose the cone/source extension is
boundary compatible, with
\[
  \Res(\phi_{\bulk})=\tau_X.
\]
Then the odd boundary pairing associated with the geometric test class is
\[
  \pair{\phi_{\bulk}}{\partial x_{\mathrm{geom}}}
  =
  \varepsilon\,\frac73,
  \qquad
  \varepsilon\in\{+1,-1\}.
\]
\end{proposition}

\begin{proof}
By boundary compatibility,
\[
  \pair{\phi_{\bulk}}{\partial x_{\mathrm{geom}}}
  =
  \varepsilon\,\pair{\tau_X}{x_{\mathrm{geom}}}.
\]
Corollary~\ref{cor:weighted-hopf-123} gives
\[
  \pair{\tau_X}{x_{\mathrm{geom}}}
  =
  \frac73.
\]
The claim follows.
\end{proof}

\subsection{Physical charge versus geometric test class}

The class
\[
  x_{\mathrm{geom}}
  =
  [T_{\bbR}X\otimes\bbC]-4[1]
\]
is a useful geometric test class. It is not automatically the physical \(C\)-field flux
class. Therefore the number
\[
  \frac73
\]
computed above is a geometric Chern--Connes value of the orbifold link
\[
  X=\Porb(1,2,3).
\]
It should not be presented as an M5 magnetic charge unless a physical argument identifies
the shifted \(C\)-field flux class with this geometric class.

For a genuine M5 charge computation on the same orbifold link, the class should be
written
\[
  x_{\check C}\in K_0(A_X),
\]
and the source charge would be
\[
  Q_{\src}(X;\check C)
  =
  \pair{\tau_X}{x_{\check C}}.
\]
At the Chern-character level, this class should satisfy
\[
  \Ch(x_{\check C})_{[4]}
  =
  \left[\frac{G_4}{2\pi}\right]_X
\]
in the appropriate orbifold, stack or groupoid cohomology theory. Its integral and
shifted quantization properties are subtler than the geometric computation above and
belong to the \(C\)-field discussion in Section~\ref{sec:cfield-anomalies}.

Thus this section establishes a controlled singular-link calculation. It shows how a
physically admissible quotient link is handled by a source-adapted Chern--Connes
pairing. It does not claim that every orbifold characteristic number is automatically an
M5-brane flux charge.

\section{Shifted \texorpdfstring{\(C\)}{C}-field quantization, anomalies and swampland constraints}
\label{sec:cfield-anomalies}

The preceding sections define source-adapted Chern--Connes pairings in smooth,
compactified and orbifold linking sectors. We now state the physical restrictions on
interpreting such a pairing as an M5-brane charge. The basic point is that the M-theory
\(C\)-field is not an arbitrary differential form. It is shifted, quantized, refined
differentially, coupled to gravity, and constrained by anomaly cancellation.

This section is therefore a consistency section. It does not change the local theorem of
Section~\ref{sec:local-m5}. Instead, it specifies what the \(K\)-class
\[
  x_{\check C}
\]
must mean if the source pairing
\[
  Q_{\src}(G_{\link};\check C)
  =
  \pair{\tau_{\link}}{x_{\check C}}
\]
is to be interpreted as the Chern--Connes shadow of an actual M-theory \(C\)-field
charge.

\subsection{The level seen by the Chern--Connes pairing}

Let
\[
  \Xeleven
\]
be an eleven-dimensional spin spacetime. Away from sources, the classical \(C\)-field
has field strength
\[
  G_4=\dd C_3.
\]
Quantum mechanically, however, the \(C\)-field is better regarded as a refined
differential cohomological object. We write it schematically as
\[
  \check C,
\]
with curvature
\[
  \operatorname{curv}(\check C)=\frac{G_4}{2\pi}.
\]
Depending on the model, \(\check C\) may be described using ordinary differential
cohomology, the \(E_8\)-gauge-theoretic model, twisted differential cohomology, or a
more speculative generalized-cohomological refinement
\cite{WittenFlux,DiaconescuMooreWittenKTheory,DiaconescuFreedMoore,HopkinsSinger,
FiorenzaSatiSchreiberE8Stack,GradySatiCohomotopy}.

The Chern--Connes pairing used in this paper does not see the full refined object
\[
  \check C.
\]
It sees its Chern-character-level image on a linking algebra. Thus, for a linking algebra
\[
  A_{\link}^{\infty}\subset A_{\link},
\]
the input class
\[
  x_{\check C}\in K_0(A_{\link})
\]
is required to satisfy
\[
  \Ch(x_{\check C})_{[4]}
  =
  \left[\frac{G_4}{2\pi}\right]_{\link}
\]
after transport to the smooth cyclic-cohomological model. In the ordinary commutative
case this reads
\[
  \ch(x_{\check C})_{[4]}
  =
  \left[\frac{G_4}{2\pi}\right]_{L}
  \in H^4_{\mathrm{dR}}(L),
\]
where \(L\) is the four-dimensional link.

Consequently
\[
  \pair{\tau_{\link}}{x_{\check C}}
\]
is the de Rham or Chern-character shadow of the \(C\)-field charge. It is not, by itself,
the full integral, torsion-sensitive or differential \(C\)-field invariant.

\begin{remark}
Complex cyclic cohomology is blind to torsion. Therefore the source pairing cannot
detect torsion fluxes, quadratic refinements, or global anomaly phases without additional
differential or generalized-cohomological data. This limitation is not incidental; it
identifies the precise level at which the present construction operates.
\end{remark}

\subsection{Shifted flux quantization}

On a spin eleven-manifold, let
\[
  \lambda(\Xeleven):=\frac12p_1(T\Xeleven)\in H^4(\Xeleven;\bbZ).
\]
Witten's flux quantization condition is
\[
  \left[\frac{G_4}{2\pi}\right]-\frac12\lambda
  \in H^4(\Xeleven;\bbZ).
\]
Equivalently, there is an integral class
\[
  a_{\check C}\in H^4(\Xeleven;\bbZ)
\]
such that
\[
  \left[\frac{G_4}{2\pi}\right]
  =
  a_{\check C}+\frac12\lambda.
\]
This gravitational shift is one of the basic quantum features of the M-theory \(C\)-field
\cite{WittenFlux}. The \(E_8\)-model and differential-cohomological formulations refine
this condition globally \cite{DiaconescuMooreWittenKTheory,DiaconescuFreedMoore,
HopkinsSinger,FiorenzaSatiSchreiberE8Stack}.

Restricting to a linking geometry \(L\), the degree-four flux class is therefore not
arbitrary. It must be the restriction of the shifted class
\[
  a_{\check C}+\frac12\lambda.
\]
Thus a \(K\)-class
\[
  x_{\check C}\in K_0(A_{\link})
\]
is physically admissible only if its Chern--Connes character represents this shifted flux
class at the degree-four level:
\[
  \Ch(x_{\check C})_{[4]}
  =
  \left(a_{\check C}+\frac12\lambda\right)\big|_{\link}.
\]
In the local product model with link
\[
  S^4,
\]
the gravitational shift is invisible in the elementary normalization calculation of
Section~\ref{sec:local-m5}. In more general compactified, orbifold or singular links, the
shift must be retained.

\begin{definition}
\label{def:cfield-admissible-class}
Let
\[
  A_{\link}^{\infty}\subset A_{\link}
\]
be a physically admissible linking algebra with cyclic four-cocycle
\[
  \tau_{\link}\in HC^4(A_{\link}^{\infty}).
\]
A class
\[
  x_{\check C}\in K_0(A_{\link})
\]
is called \emph{\(C\)-field-admissible} if its Chern--Connes character represents the
degree-four de Rham image of the shifted \(C\)-field flux on the link:
\[
  \Ch(x_{\check C})_{[4]}
  =
  \left[\frac{G_4}{2\pi}\right]_{\link}
  =
  \left(a_{\check C}+\frac12\lambda\right)\big|_{\link}.
\]
\end{definition}

With this terminology, the source charge associated with a \(C\)-field-admissible class is
\[
  Q_{\src}(G_{\link};\check C)
  =
  \pair{\tau_{\link}}{x_{\check C}}.
\]

\subsection{M5 magnetic sources}

Let
\[
  \Wsix\hookrightarrow \Xeleven
\]
be an M5-brane worldvolume, and let
\[
  \calN:=N_{W/X}\longrightarrow W
\]
be the rank-five normal bundle. In the presence of the M5-brane, the \(C\)-field is
singular along \(W\), and one works on the complement or on the boundary of a tubular
neighbourhood of \(W\). Locally, the source Bianchi identity has the current form
\[
  \dd G_4=2\pi N\,\delta_W,
\]
where \(N\) is the magnetic charge. For a small linking four-cycle
\[
  L\simeq S^4
\]
in the normal fibre, the charge is
\[
  Q_{\Mfive}(L)
  =
  \int_L\frac{G_4}{2\pi}.
\]

The source-pairing replacement is
\[
  Q_{\Mfive}(L)
  \quad\leadsto\quad
  Q_{\src}(G_{\link};\check C)
  =
  \pair{\tau_{\link}}{x_{\check C}}.
\]
Here the cyclic cocycle
\[
  \tau_{\link}
\]
represents the linking fundamental class, and
\[
  x_{\check C}
\]
represents the shifted \(C\)-field flux sector at the Chern-character level.

\begin{proposition}
\label{prop:cfield-shadow-charge}
Let \(L\) be a smooth compact oriented four-dimensional link, and let
\[
  \tau_L\in HC^4(C^\infty(L))
\]
be the normalized fundamental cyclic cocycle. Suppose
\[
  x_{\check C}\in K^0(L)
\]
satisfies
\[
  \ch(x_{\check C})_{[4]}
  =
  \left[\frac{G_4}{2\pi}\right]_L.
\]
Then
\[
  \pair{\tau_L}{x_{\check C}}
  =
  \int_L\frac{G_4}{2\pi}.
\]
\end{proposition}

\begin{proof}
By the normalization of \(\tau_L\),
\[
  \pair{\tau_L}{x_{\check C}}
  =
  \int_L\ch(x_{\check C})_{[4]}.
\]
Using
\[
  \ch(x_{\check C})_{[4]}
  =
  \left[\frac{G_4}{2\pi}\right]_L,
\]
we obtain
\[
  \pair{\tau_L}{x_{\check C}}
  =
  \int_L\frac{G_4}{2\pi}.
\]
\end{proof}

This proposition is elementary, but it records the exact physical level of the
construction. The Chern--Connes pairing computes the flux integral once the \(K\)-class
has the correct Chern-character image.

\subsection{Anomaly inflow and the limit of the additive charge}

The M5-brane worldvolume theory is chiral. It contains a self-dual tensor field, chiral
fermions and normal-bundle degrees of freedom. Therefore its anomaly structure is highly
constrained. Witten analyzed the fivebrane effective action and the self-dual field
\cite{WittenFivebrane}; Freed, Harvey, Minasian and Moore showed that gravitational
anomaly cancellation for the M-theory fivebrane requires the eleven-dimensional
Chern--Simons coupling and its modification in the presence of the fivebrane source
\cite{FreedHarveyMinasianMoore}. Further refinements of self-dual-field and global
anomaly questions appear in \cite{MooreWitten,BelovMooreSelfDual,MonnierGlobalM5}.

Let
\[
  \calN\to W
\]
be the rank-five normal bundle of the M5-brane. With standard orientation conventions,
the anomaly polynomial for a single M5-brane contains
\[
  I_8^{(1)}(TW,\calN)
  =
  \frac{1}{48}
  \left[
    p_2(\calN)-p_2(TW)
    +
    \frac14\bigl(p_1(TW)-p_1(\calN)\bigr)^2
  \right].
\]
For \(r\) coincident M5-branes of type \(A_{r-1}\), the anomaly polynomial contains the
well-known cubic contribution
\[
  I_8^{A_{r-1}}
  =
  r\,I_8^{(1)}
  +
  \frac{r^3-r}{24}\,p_2(\calN),
\]
up to the convention for including the center-of-mass tensor multiplet
\cite{HarveyMinasianMooreTensor,IntriligatorHopfWZ}.

This has a direct consequence for the present construction. The source pairing
\[
  \pair{\tau_{\link}}{x_{\check C}}
\]
is additive in the \(K\)-class. It is well suited to measuring a topological linking charge.
By itself it cannot reproduce the full nonabelian \(r^3\)-scaling of coincident M5-brane
anomalies. Any attempt to obtain the full anomaly polynomial would require additional
nonabelian, secondary, differential or anomaly-field-theoretic structure.

Thus the claim of this paper is restricted:
\begin{center}
\fbox{%
\begin{minipage}{0.78\textwidth}
\centering
The source-adapted Chern--Connes pairing captures a linking magnetic charge,\\
not the full M5-brane anomaly polynomial.
\end{minipage}}
\end{center}

\subsection{Compatibility with the eleven-dimensional topological term}

The eleven-dimensional supergravity action contains a Chern--Simons term. In
topologically nontrivial backgrounds, the topological part is most cleanly understood
through a twelve-dimensional extension or a differential-cohomological formulation.
At the differential-form level one writes schematically
\[
  2\pi i
  \int_Z
  \left(
    \frac16\,G_4\wedge G_4\wedge G_4
    -
    G_4\wedge I_8
  \right),
  \qquad
  \partial Z=\Xeleven,
\]
where
\[
  I_8=
  \frac{1}{48}
  \left(
    p_2-\frac14p_1^2
  \right),
\]
with the appropriate tangent-bundle convention. The precise quantum expression involves
the shifted \(C\)-field, integral lifts and quadratic refinements
\cite{WittenFlux,DiaconescuFreedMoore,MoorePageCharges,HopkinsSinger}.

For an M5-brane source, the variation of this topological term contributes anomaly
inflow onto the fivebrane. Therefore a source-adapted Chern--Connes charge must be
compatible with anomaly inflow. It is not enough to produce a number; the class
\[
  x_{\check C}
\]
must be compatible with shifted flux quantization and with the normal-bundle and
worldvolume data entering the M5 anomaly cancellation mechanism.

\begin{definition}
\label{def:anomaly-compatible-source-datum}
A source datum
\[
  (A_{\link}^{\infty}\subset A_{\link},\tau_{\link},x_{\check C})
\]
is called \emph{anomaly-compatible} if:
\begin{enumerate}[label=(\roman*)]
  \item \(x_{\check C}\) is \(C\)-field-admissible in the sense of
  Definition~\ref{def:cfield-admissible-class};

  \item the corresponding boundary class
  \[
    \partial x_{\check C}\in K_1(A_{\bulk})
  \]
  satisfies the boundary compatibility identity
  \[
    \pair{\phi_{\bulk}}{\partial x_{\check C}}
    =
    \varepsilon\,
    \pair{\tau_{\link}}{x_{\check C}};
  \]

  \item the resulting charge is compatible with the shifted flux quantization and the
  M5-brane anomaly-inflow constraints of the background.
\end{enumerate}
\end{definition}

The third condition is intentionally not compressed into a single formula. In a full
M5-brane background it includes the normal \(SO(5)\)-bundle, the self-dual tensor field,
the shifted \(C\)-field, quadratic refinements and possible torsion data.

\subsection{Torsion and differential refinements}

The source pairing
\[
  \pair{\tau_{\link}}{x_{\check C}}
\]
is a complex number obtained from the Chern--Connes character. It therefore belongs to
the last stage of a hierarchy:
\[
  \begin{aligned}
  \text{differential or generalized cohomology}
  &\longrightarrow \text{integral characteristic class}\\
  &\longrightarrow \text{Chern character}
  \longrightarrow \text{cyclic pairing}.
  \end{aligned}
\]
The first stages contain information lost by the last stage.

Torsion fluxes, Heisenberg groups of fluxes, quadratic refinements and global anomaly
phases are essential in the quantization of self-dual fields and generalized abelian gauge
fields \cite{FreedDiracCharge,FreedMooreSegalFlux,FreedMooreSegalHeisenberg,
HopkinsSinger,BelovMooreSelfDual}. Thus, in a fully refined treatment, one should
replace the bare class
\[
  x_{\check C}\in K_0(A_{\link})
\]
by a refined object
\[
  \widehat{x}_{\check C}
\]
in a suitable differential or generalized cohomology theory, together with maps
\[
  \widehat{x}_{\check C}
  \longmapsto
  x_{\check C}
  \longmapsto
  \Ch(x_{\check C}).
\]
The Chern--Connes pairing is then the numerical shadow of this refined datum.

The present paper does not choose a unique refined model. Ordinary differential
cohomology, the \(E_8\)-model, and cohomotopy-based proposals each capture different
aspects of the \(C\)-field problem
\cite{DiaconescuFreedMoore,HopkinsSinger,FiorenzaSatiSchreiberE8Stack,
FiorenzaSatiSchreiberCohomotopy,SatiSchreiberM5,GradySatiCohomotopy}.
The result proved here is conditional: whenever such a refined \(C\)-field datum has a
Chern-character-level image represented by a \(K\)-class on the linking algebra, the
Chern--Connes pairing computes the corresponding topological linking charge.

\subsection{Swampland constraints}

The swampland program asks which apparently consistent effective theories can arise in
quantum gravity. Vafa introduced the landscape/swampland distinction in
\cite{VafaSwampland}. The distance conjecture and related geometric constraints appear
in \cite{OoguriVafaDistance}. The Weak Gravity Conjecture was formulated in
\cite{ArkaniHamedMotlNicolisVafaWGC}, and modern reviews include
\cite{PaltiReview,HarlowHeidenreichReeceRudeliusWGC}. For the present paper, the
most relevant themes are charge, global symmetries and cobordism.

First, a source charge in quantum gravity should not be interpreted as an exact global
charge unattached to dynamical objects. The absence of global symmetries in quantum
gravity is emphasized in \cite{BanksSeiberg,HarlowOoguriConstraints,
HarlowOoguriSymmetries}. Therefore
\[
  Q_{\src}(G_{\link};\check C)
  =
  \pair{\tau_{\link}}{x_{\check C}}
\]
should not be promoted to a new global superselection label. It should be interpreted as
a gauge/topological charge measured by a linking geometry around a dynamical brane
source.

Second, the Weak Gravity Conjecture and its higher-form generalizations constrain
allowed charged objects in quantum gravity. The Chern--Connes pairing computes only
the topological charge. It does not test a tension-to-charge inequality or guarantee the
existence of a sufficiently light charged brane. Such conditions must be checked
separately in any full compactification or effective theory.

Third, the cobordism swampland conjecture of McNamara and Vafa states, in broad
terms, that nontrivial cobordism classes in quantum gravity should be trivialized by
dynamical defects \cite{McNamaraVafaCobordism}. This is directly relevant to a
linking-charge construction. If a source-adapted Chern--Connes pairing defines a
nontrivial bordism invariant of links or source sectors, then a quantum-gravity
interpretation requires one of two outcomes:
\begin{enumerate}[label=(\roman*)]
  \item the invariant is the charge of a dynamical brane/source sector; or

  \item the corresponding cobordism class is trivialized by allowed defects, domain walls
  or brane nucleation processes.
\end{enumerate}
Otherwise the construction risks producing precisely the kind of absolute topological
label excluded by swampland reasoning.

The swampland-compatible interpretation of the present construction is therefore:
\begin{center}
\fbox{%
\begin{minipage}{0.86\textwidth}
\centering
The source-adapted Chern--Connes pairing may represent the Chern-character shadow\\
of an allowed brane/source charge, but it should not be interpreted as an independent\\
global charge or an unrestricted effective-field-theory parameter.
\end{minipage}}
\end{center}

\subsection{Admissibility checklist}

We summarize the physical admissibility requirements for a proposed M5-brane source
pairing
\[
  Q_{\src}(G_{\link};\check C)
  =
  \pair{\tau_{\link}}{x_{\check C}}.
\]

\begin{enumerate}[label=(\roman*)]
  \item \textbf{Physical linking algebra.}
  The algebra \(A_{\link}\) must arise from a smooth link, CDS compactification,
  \(B\)- or \(C\)-field brane limit, orbifold quotient, or another explicit source
  geometry. It must not be an arbitrary transverse groupoid.

  \item \textbf{\(C\)-field admissibility.}
  The class \(x_{\check C}\) must represent the shifted flux class
  \[
    \left[\frac{G_4}{2\pi}\right]
    =
    a_{\check C}+\frac12\lambda
  \]
  at the Chern-character level on the link.

  \item \textbf{Boundary compatibility.}
  The even linking pairing and the odd codimension-five boundary pairing must agree:
  \[
    \pair{\phi_{\bulk}}{\partial x_{\check C}}
    =
    \varepsilon\,
    \pair{\tau_{\link}}{x_{\check C}}.
  \]

  \item \textbf{Anomaly compatibility.}
  The source charge must be compatible with M5 anomaly inflow, including the normal
  \(SO(5)\)-bundle and the shifted \(C\)-field.

  \item \textbf{Differential refinement.}
  Torsion, quadratic refinements and global anomaly phases require a refined object
  \[
    \widehat{x}_{\check C}
  \]
  beyond the complex cyclic pairing.

  \item \textbf{Swampland compatibility.}
  The charge must be attached to dynamical brane/source sectors or be trivialized by
  allowed defects. It must not define an exact global topological superselection label.
\end{enumerate}

The local theorem of Section~\ref{sec:local-m5} verifies the first normalization test:
in the smooth product model the Chern--Connes pairing gives the ordinary M5 magnetic
charge. The remaining requirements are the conditions for promoting the formula from a
topological linking computation to a physically admissible M-theory source charge.

\section{Conclusion and outlook}
\label{sec:conclusion}

We have formulated a source-adapted Chern--Connes expression for M5-brane magnetic
charge in physically motivated linking and compactification sectors. The construction
starts from the ordinary local geometry of an M5-brane:
\[
  \Wsix\hookrightarrow \Xeleven,
  \qquad
  \rank N_{W/X}=5.
\]
A magnetic source of codimension five is measured by a four-dimensional linking cycle.
In the smooth local model the link is
\[
  S^4_{\link},
\]
and the physical charge is
\[
  Q_{\Mfive}
  =
  \int_{S^4_{\link}}\frac{G_4}{2\pi}.
\]
The main point of the paper is that this linking charge can be written as a
Chern--Connes pairing once the \(C\)-field flux sector is represented by an appropriate
\(K\)-theory class on the linking algebra.

The published NCFI construction in \cite{ZoisSigma} provides the mathematical
prototype: a cyclic fundamental class is paired with a \(K\)-theory class. In the published
even-codimensional foliation case, the cyclic class is Connes' transverse fundamental
cyclic cocycle and the \(K_0\)-class comes from the transverse geometry of the foliation
\cite{ConnesTransverse,ConnesNCG}. The M5-brane problem considered here is different.
It is a codimension-five source problem. The relevant \(K\)-class is not a transverse
geometric module class; it is a charge class associated with the \(C\)-field flux on the
link.

This distinction is central. In the present paper the physically transparent object is the
even linking pairing
\[
  Q_{\src}(A_{\link};\check C)
  =
  \pair{\tau_{\link}}{x_{\check C}},
\]
where
\[
  \tau_{\link}\in HC^4(A_{\link}^{\infty})
\]
is the cyclic fundamental cocycle of the four-dimensional link or compactification sector,
and
\[
  x_{\check C}\in K_0(A_{\link})
\]
is the \(K\)-theoretic charge class whose Chern--Connes character represents the
degree-four part of the restricted \(C\)-field flux. The odd codimension-five class is then
not chosen independently. It is obtained from the source boundary map:
\[
  [u_{\check C}]
  =
  \partial x_{\check C}
  \in K_1(A_{\bulk}).
\]
Thus the strict odd expression
\[
  Z_{\mathrm{bdry}}(A_{\bulk};\check C)
  =
  \pair{\phi_{\bulk}}{\partial x_{\check C}}
\]
is the boundary form of the same charge, provided the cyclic residue and the
\(K\)-theory boundary map satisfy the compatibility identity
\[
  \pair{\phi_{\bulk}}{\partial x_{\check C}}
  =
  \varepsilon\,\pair{\tau_{\link}}{x_{\check C}}.
\]

The local normalization theorem proves the basic case. Let \(E_N\to S^4\) be a rank-two
\(SU(2)\)-bundle with
\[
  \int_{S^4}c_2(E_N)=N.
\]
With the orientation convention used in Section~\ref{sec:local-m5}, define
\[
  x_N:=2[1]-[E_N]\in K_0(C(S^4)).
\]
Then
\[
  \pair{\tau_{S^4}}{x_N}=N.
\]
If
\[
  \left[\frac{G_4}{2\pi}\right]_{S^4}=c_2(E_N),
\]
this becomes
\[
  \pair{\tau_{S^4}}{x_N}
  =
  \int_{S^4}\frac{G_4}{2\pi}.
\]
This result is deliberately elementary. It is the normalization check which any
noncommutative, compactified or quotient version must reproduce in its commutative
limit.

We then explained why noncommutative and groupoid algebras enter the article only
when they have a physical origin. The first mechanism is the Connes--Douglas--Schwarz
appearance of noncommutative tori in Matrix/M-theory toroidal compactification with
background three-form data \cite{ConnesDouglasSchwarz}. In that setting a
noncommutative four-torus
\[
  A_\Theta=C(\bbT^4_\Theta)
\]
carries a fundamental cyclic four-cocycle
\[
  \tau_\Theta^{(4)}\in HC^4(A_\Theta^\infty).
\]
For a compactified charge class
\[
  x_{\check C,\Theta}\in K_0(A_\Theta),
\]
the compactified source charge is
\[
  Q_{\src,\Theta}
  =
  \pair{\tau_\Theta^{(4)}}{x_{\check C,\Theta}}.
\]
This pairing has the correct commutative limit and is invariant under Morita equivalence
when the cyclic cocycle and the \(K\)-class are transported through the same equivalence
bimodule. This gives the expected duality-compatibility condition for CDS-type
compactification sectors \cite{SchwarzMorita,RieffelSchwarz}.

The second physical source of noncommutativity is reduction to D-brane systems with
background \(B\)-field, together with the related M5-brane large \(C\)-field and
open-membrane limits. D-branes in background \(B\)-fields lead to noncommutative
worldvolume descriptions \cite{DouglasHull,SeibergWitten,Schomerus}. M5-branes in
large \(C\)-field backgrounds lead to open-membrane, noncommutative or
Nambu--Poisson-type structures \cite{BergshoeffNoncommM5,HoLargeCFieldReview,
ChenFuruuchiHoTakimi}. The associative Chern--Connes pairing used here should
therefore be interpreted as the associative or decategorified charge shadow available in
those regimes, not as a complete description of the full higher-algebraic M5-brane
dynamics.

The third physically admissible source of groupoid algebras is quotient geometry.
Orbifold and quotient singularities are standard in brane physics, including fivebranes
on orbifolds, D-branes at ALE singularities and M2-branes probing quotient singularities
\cite{WittenOrbifoldFivebranes,DouglasMooreOrbifolds,ABJM}. In such cases the naive
quotient space loses stabilizer data. The groupoid or stack retains it. This justifies the
use of action groupoids and orbifold \(C^*\)-algebras for quotient links.

The weighted projective example illustrates this groupoid mechanism. The weighted
Hopf action
\[
  \lambda\cdot(z_0,z_1,z_2)
  =
  (\lambda z_0,\lambda^2z_1,\lambda^3z_2)
\]
presents the weighted projective orbifold
\[
  X=\Porb(1,2,3)
\]
as a four-dimensional link. For the geometric test class
\[
  x_{\mathrm{geom}}
  =
  [T_{\bbR}X\otimes\bbC]-4[1],
\]
the source-adapted Chern--Connes pairing becomes
\[
  \pair{\tau_X}{x_{\mathrm{geom}}}
  =
  \int_X^{\orb}p_1(T_{\bbR}X).
\]
Using the weighted Euler sequence and the normalization
\[
  \int_{\Porb(1,2,3)}^{\orb}H^2=\frac16,
\]
we obtained
\[
  \pair{\tau_X}{x_{\mathrm{geom}}}
  =
  \frac73.
\]
This is a controlled orbifold characteristic-number calculation. It is not automatically
an M5-brane magnetic charge. A physical M5 interpretation would require replacing
\(x_{\mathrm{geom}}\) by a \(C\)-field-admissible charge class
\[
  x_{\check C}\in K_0(A_X)
\]
whose Chern--Connes character represents the shifted \(G_4\)-flux on the orbifold link.

The \(C\)-field discussion makes this restriction precise. The source pairing operates at
the Chern-character level:
\[
  \Ch(x_{\check C})_{[4]}
  =
  \left[\frac{G_4}{2\pi}\right]_{\link}.
\]
But in M-theory the flux is shifted:
\[
  \left[\frac{G_4}{2\pi}\right]-\frac12\lambda
  \in H^4(-;\bbZ),
  \qquad
  \lambda=\frac12p_1.
\]
Thus the \(K\)-class used in the pairing must be the Chern--Connes shadow of a refined
\(C\)-field datum, not an arbitrary \(K_0\)-class. Differential cohomology, the
\(E_8\)-model and generalized-cohomological refinements contain information that the
complex cyclic pairing does not see \cite{WittenFlux,DiaconescuFreedMoore,
HopkinsSinger,FiorenzaSatiSchreiberE8Stack,GradySatiCohomotopy}.

The same caution applies to anomalies. The Chern--Connes pairing is additive in the
\(K\)-class and is well suited to measuring a topological linking charge. It does not by
itself reproduce the full anomaly polynomial of the six-dimensional \((2,0)\) theory, nor
the \(r^3\)-scaling of coincident M5-branes. Those structures require the self-dual tensor
field, the rank-five normal bundle, Wess--Zumino terms, anomaly inflow and refined
global data \cite{WittenFivebrane,FreedHarveyMinasianMoore,
HarveyMinasianMooreTensor,IntriligatorHopfWZ,MonnierGlobalM5}. The result of the
present paper should therefore be stated as a linking-charge result, not as a derivation of
M5-brane dynamics.

The swampland discussion imposes a further interpretive boundary. A source charge in
quantum gravity should not be promoted to a new exact global superselection label. It
should be tied to dynamical brane/source sectors or be trivialized by allowed defects or
cobordisms. This is consistent with the absence of global symmetries in quantum gravity,
the Weak Gravity Conjecture and the cobordism swampland viewpoint
\cite{VafaSwampland,ArkaniHamedMotlNicolisVafaWGC,BanksSeiberg,
HarlowOoguriConstraints,HarlowOoguriSymmetries,McNamaraVafaCobordism}. Thus the
source-adapted Chern--Connes charge should be read as a possible cyclic shadow of an
allowed brane charge, not as an unrestricted effective-field-theory parameter.

The paper establishes the following results.

\begin{enumerate}[label=(\roman*)]
  \item It distinguishes the published even-codimensional NCFI \cite{ZoisSigma} from
  the codimension-five source problem. The former pairs Connes' transverse fundamental
  cyclic cocycle with a transverse-geometric \(K_0\)-class. The latter uses a
  \(C\)-field charge class on a four-dimensional link.

  \item It identifies the boundary/source extension as the natural origin of the required
  odd \(K_1\)-class in codimension five:
  \[
    [u_{\check C}]=\partial x_{\check C}.
  \]

  \item It proves the local smooth \(S^4\) normalization theorem:
  \[
    \pair{\tau_{S^4}}{x_N}
    =
    \int_{S^4}\frac{G_4}{2\pi}
  \]
  under the stated identification of the \(K\)-class with the flux class.

  \item It formulates the CDS compactified version over a noncommutative four-torus:
  \[
    Q_{\src,\Theta}
    =
    \pair{\tau_\Theta^{(4)}}{x_{\check C,\Theta}},
  \]
  with Morita invariance as the duality-compatibility condition.

  \item It computes a weighted projective orbifold link example:
  \[
    \pair{\tau_{\Porb(1,2,3)}}{[T_{\bbR}\Porb(1,2,3)\otimes\bbC]-4[1]}
    =
    \frac73.
  \]

  \item It states the physical admissibility conditions imposed by shifted \(C\)-field
  quantization, anomaly inflow, differential refinements and swampland constraints.
\end{enumerate}

Several problems remain open.

First, the boundary compatibility identity
\[
  \pair{\phi_{\bulk}}{\partial x}
  =
  \varepsilon\,\pair{\tau_{\link}}{x}
\]
should be proved in a wider class of source and cone groupoids. The local commutative
extension is clear, and the formal cyclic-excision mechanism is standard, but a fully
general source-groupoid theorem would require a careful comparison between smooth
cyclic models, \(K\)-theory boundary maps and groupoid cone extensions.

Second, the \(C\)-field-admissibility condition should be lifted from the Chern-character
level to a genuinely integral and differential statement. One should construct a refined
object
\[
  \widehat{x}_{\check C}
\]
in an appropriate differential or generalized cohomology theory, together with a map to
the \(K\)-class
\[
  x_{\check C}
\]
used in the cyclic pairing. This is necessary to incorporate torsion, quadratic refinements
and global anomaly phases.

Third, the orbifold example should be extended from the geometric test class
\[
  x_{\mathrm{geom}}
\]
to actual shifted \(C\)-field flux classes on orbifold or stack links. This requires a
careful treatment of stack cohomology, equivariant cohomology, twisted sectors and
possible torsion corrections. Only after that can one interpret an orbifold pairing as a
physical M5 source charge.

Fourth, the noncommutative-torus compactification formula should be tested on explicit
projective modules over \(A_\Theta\). The resulting numbers
\[
  \pair{\tau_\Theta^{(4)}}{x_\Theta}
\]
should be compared with known Chern numbers of noncommutative tori, with
Morita-duality transformations and with the \(C\)-field data appearing in CDS-type
compactifications.

Fifth, the anomaly-inflow relation should be developed beyond charge matching. A
future refinement would have to identify secondary cyclic or differential characters
associated with Wess--Zumino or anomaly-inflow terms, and explain how the additive
linking charge interacts with the nonlinear and nonabelian structure of coincident
fivebranes.

The conservative conclusion is therefore this. The construction proposed here is not a
candidate for the full M-theory Lagrangian and not a derivation of the six-dimensional
\((2,0)\) theory. Its natural role is narrower: it is a Chern--Connes expression for
topological source charges in M5-brane sectors whose linking or compactification
geometry is smooth, noncommutative-toroidal, or quotient/groupoid in a physically
admissible way. In the ordinary local model it reproduces the M5 magnetic charge. In
CDS compactified sectors and orbifold source links it gives a controlled extension of the
same charge-pairing principle.

\end{document}